\documentclass[12pt]{article}         
\usepackage{amssymb,amsmath,amsbsy,amsfonts}
\usepackage{cite}
\usepackage{epsfig}
\usepackage{graphics}
\usepackage{epsf}
\usepackage[pdfborder={0 0 0}, colorlinks=false, breaklinks=false]{hyperref}

\newlength{\dinwidth} 
\newlength{\dinmargin}
\newtheorem{theorem}{Theorem}[section]
\newtheorem{prop}[theorem]{Proposition}
\newtheorem{lemma}[theorem]{Lemma}

\newtheorem{definition}[theorem]{Definition}
 
\newenvironment{proof}{\medskip \noindent 
            {\bf Proof.}}{ \hfill $\square$ \medskip}

\newcommand{\ie}{{\it i.e.\ }}

\newcommand{\cf}{{\it cf.\ }}

\newcommand{\RR}{\mathbb R}
\newcommand{\CC}{\mathbb C}
\newcommand{\NN}{\mathbb N}
\newcommand{\ZZ}{\mathbb Z}

\newcommand{\BH}{{\cal B(H)}}

\newcommand{\C}{{\cal C}}

\newcommand{\tq}{\times_Q}
\newcommand{\tQ}{\times^Q}

\newcommand{\RC}{\boldsymbol{\cal C}}
\newcommand{\IC}{\boldsymbol{\cal C}^\infty}

\newcommand{\ra}{\boldsymbol A}
\newcommand{\rb}{\boldsymbol B}

\newcommand{\rtm}{\boldsymbol{\times}}

\def\idty{{\leavevmode\hbox{\rm 1\kern -.3em I}}}

\def\As{{\cal A}}
\def\Bs{{\cal B}}
\def\Cs{{\cal C}}
\def\Ds{{\cal D}}

\def\Fs{{\cal F}}

\def\Hs{{\cal H}}

\def\Ls{{\cal L}}

\def\Ns{{\cal N}}

\def\Ps{{\cal P}}

\def\Rs{{\cal R}}
\def\Ss{{\cal S}}
\def\Ts{{\cal T}}

\def\Ws{{\cal W}}

\def\Pid{{\Ps_+ ^{\uparrow}}}
\def\Lid{{\Ls_+ ^{\uparrow}}}

\def\idty{{\leavevmode\hbox{\rm 1\kern -.3em I}}}

\def\RR{{\mathbb R}}
\def\CC{{\mathbb C}}
\def\NN{{\mathbb N}}

\def\ZZ{{\mathbb Z}}

\def\rest{\upharpoonright}

\def\cA{{\cal A}}

\def\cC{{\cal C}}

\def\cH{{\cal H}}

\def\cO{{\cal O}}

\def\cR{{\cal R}}

\def\cW{{\cal W}}

\def\CC{{\mathbb C}}

\def\NN{{\mathbb N}}

\def\RR{{\mathbb R}}

\def\ZZ{{\mathbb Z}}

\begin{document}

\title{\Large \bf Warped Convolutions, Rieffel Deformations \\
and the Construction of Quantum Field Theories}

\author{\large { Detlev Buchholz\,$^a$\thanks{Supported by the German 
Research Foundation (Deutsche Forschungsgemeinschaft (DFG)) through the 
Institutional Strategy of the University of G\"ottingen}, 
\ Gandalf Lechner\,$^b$  and \
Stephen J.\ Summers\,$^c$\thanks{Research supported by the NSF Grant DMS-0901370.} }\\[4mm]
${}^a$ Institut f\"ur Theoretische Physik and 
Courant Centre \\
``Higher Order Structures in Mathematics'', 
Universit\"at G\"ottingen, \\ 37077 G\"ottingen, Germany  \\[2mm]
${}^b$ Fakult\"at f\"ur Physik, 
Universit\"at Wien, \\ 1090 Vienna, Austria  \\[2mm]
${}^c$ Department of Mathematics, 
University of Florida, \\ Gainesville FL 32611, USA}

\date{\large }

\maketitle 

\begin{abstract}  \noindent
Warped convolutions of operators were recently introduced in the
algebraic framework of quantum physics as a new constructive tool. 
It is shown here that these convolutions provide isometric 
representations of Rieffel's strict deformations of $C^*$--dynamical
systems with automorphic actions of~$\RR^n$, whenever the latter are 
presented in a covariant representation. Moreover, the device can 
be used for the deformation of relativistic quantum field theories
by adjusting the convolutions to the geometry of Minkowski space.
The resulting deformed theories still comply with pertinent physical
principles and their Tomita--Takesaki modular data coincide with those
of the undeformed theory; but they are in general inequivalent to the
undeformed theory and exhibit different physical 
interpretations.

\end{abstract}

\newpage

\section{Introduction}

     Recent advances in algebraic quantum field theory have led to
purely algebraic constructions of quantum field models on Minkowski
space, both classical and noncommutative
\cite{BrGuLo2,Sch,Le,BuLe,Le2,MuSchYng,BuSuads,BuSu1,Le3,GrLe,GrLe2,BuSu2}, 
many of which cannot be achieved by the standard methods of constructive
quantum field theory. Some of these models are local and free, some
are local and have nontrivial $S$--matrices, and yet others manifest only
certain remnants of locality, though these remnants suffice to enable
the computation of nontrivial $S$--matrix elements.

     In order to construct a quantum field model on noncommutative 
Minkowski space, Grosse and one of us \cite{GrLe} have deformed
the free quantum field in a certain manner to find a family of 
theories which are Poincar\'e covariant and comply with a slightly weakened 
version of the principle of Einstein causality  (``wedge locality''). 
As pointed out in \cite{GrLe}, a completely analogous deformation can be 
carried out on a free field on classical
Minkowski space. In \cite{BuSu2} two of us presented a generalization (called 
a warped convolution) of that deformation which can be applied to any 
Minkowski space quantum field model in any number of dimensions. This 
deformation results in a family of distinct theories  
which are wedge--local and covariant under the representation of the 
Poincar\'e group associated with the initial, undeformed theory. It turns out 
that also the $S$--matrix changes under this deformation, and the scattering
is nontrivial even if the scattering of the initial theory is trivial.  
When taking the free quantum field as the initial
model, this deformation coincides with that of Grosse and Lechner. It provides the first 
fully consistent examples of relativistic quantum field theories on
four--dimensional Minkowski space describing nontrivial elastic 
scattering processes \cite{GrLe,BuSu2}. 

    Warped convolution was subsequently studied in the language of Wightman
quantum field theory in \cite{GrLe2}, where it was shown that the 
deformation of the field operators can be understood as resulting in a 
certain deformation of the canonical product on the Borchers--Uhlmann algebra 
-- the algebra of test functions canonically associated with a Wightman theory. 
This was the first indication that the deformation of operators resulting
from warped convolution may be equivalent to a deformation of the 
operator product. 

    A well known example of this latter type is the strict deformation 
theory of $C^*$--dynamical systems with an action of $\RR^n$
developed by Rieffel in \cite{Ri}. It was originally introduced
for the quantization of classical models. We shall show in this
paper that the warped convolution applied to  
any $C^*$--dynamical system provides a 
covariant representation of the corresponding deformed
Rieffel algebra. In particular, all states in a covariant representation of 
the initial system can be lifted to states on the deformed algebra
(compare \cite[Corollary 4.4]{KNW}). 

     In spite of this tight relation between the two deformation 
procedures, the concept of warped convolution appears to be more
appropriate in applications to quantum field theory. For there one has 
to deal simultaneously with a multitude of different deformations and to 
establish relations between the resulting 
\mbox{operators}. This can be done 
most conveniently in a common representation space of the various Rieffel 
algebras, and such a space is provided by the warped convolution procedure. 
Suitably adjusting the deformation parameters to the geometry of Minkowski 
space, we apply the warped convolutions to quantum field theories, as 
outlined in \cite{BuSu2}, and prove as well as extend the results given there. 
As we shall further explain, any quantum field theory on Minkowski space
can be constructed from a (causal) Borchers triple consisting of a 
von Neumann algebra, a representation of the Poincar\'e group and a vector 
representing the vacuum state. The physical constraints of causality and 
covariance can conveniently be expressed in terms of a few conditions on 
these triples. We shall show that these properties are preserved under a 
distinguished group of warped convolutions, thereby giving rise to interesting 
new theories.

     The article is organized as follows. In Section \ref{strict} we prove 
and extend the results about warped convolution given in \cite{BuSu2}. These 
extensions allow us to establish the relation with Rieffel deformed 
dynamical systems. In Section \ref{triple} a restricted family of 
warped convolutions is applied to Borchers triples to 
construct quantum field theories in two spacetime dimensions. We show, in 
particular, that the Tomita--Takesaki modular objects associated with such 
triples remain fixed under these deformations. The application of the warped 
convolutions to general relativistic quantum field theories in higher 
dimensions is discussed in Section \ref{application}. We present there the 
salient results given in \cite{BuSu2} in the framework of causal Borchers 
triples and also establish further physically relevant properties of 
the deformed theories, not addressed in \cite{BuSu2}. Finally, we 
briefly discuss prospects for further development of this approach in 
Section \ref{conclusions}.

\section{\hspace*{-6mm} 
Warped Convolutions  and Rieffel Deformations} \label{strict}
\setcounter{equation}{0} 

     We clarify here the relation between the notion of warped convolution, 
recently introduced in \cite{BuSu2}, and the strict deformation of 
$C^*$--algebras established in \cite{Ri} by Rieffel. In either case one 
proceeds from a $C^*$--dynamical system $(\As,\RR^n)$,  \cf \cite{Pe}. 
It consists of a $C^*$--algebra $\As$ equipped  
with a strongly continuous automorphic action of the group $\RR^n$
which will be denoted by $\alpha$.

     In order to relate the two settings, it will be convenient to 
consider the system $(\As,\RR^n)$ in a covariant representation. 
That is, we regard $\As$ as a concrete $C^*$--algebra on a Hilbert 
space $\Hs$ on which the automorphisms $\alpha$ are implemented by the
adjoint action of a weakly continuous unitary representation $U$ of
$\RR^n$,
$$ \alpha_x(A) = U(x) A U(x)^{-1} \, , \quad x \in \RR^n \, . $$
As a matter of fact, this assumption imposes no significant restriction of 
generality. For if the abstract algebra $\As$ can be represented faithfully on
some separable Hilbert space, then there also exists a faithful covariant 
representation of $(\As,\RR^n)$, \cf 
\cite[Lemma 7.4.9 and Prop. 7.4.7]{Pe}. 
Furthermore, since the adjoint action 
$\alpha$ of the unitary representation $U$ can be extended to the algebra 
$\Bs(\Hs)$ of all bounded operators on $\Hs$, no generality will be lost 
when we proceed to the $C^*$--dynamical system $(\Cs,\RR^n)$, where 
$\Cs \subset \Bs(\Hs)$ is the $C^*$--algebra of all operators on which 
$\alpha$ acts strongly continuously. We shall then be able to restrict to 
suitable subalgebras $\As$ of $\Cs$ as necessary.

\subsection{Rieffel Deformations} 

     We begin by considering the $C^*$--algebra $\RC$ of all  
uniformly continuous 
bounded functions $\ra: \, \RR^n \rightarrow {\cal B(H)}$. The 
algebraic structure of $\RC$ is the natural one inherited from ${\cal B(H)}$, 
\ie the algebraic operations in $\RC$ are pointwise defined,
$$ (\ra + \rb)(x) = \ra(x) + \rb(x), \
(\ra \, \rb)(x) = \ra(x) \, \rb(x), \ \ra^*(x) = \ra(x)^* \quad
x \in \RR^n \, ,   $$
and the norm is given by\footnote{Risking some confusion,
we use the same symbol for the norms on $\RC$ and ${\cal B(H)}$.}
$$ || \ra || = \sup_{ x \in \RR^n} ||\ra(x)|| \, .$$
Following Rieffel \cite{Ri}, we consider the subalgebra $\IC \subset \RC$ 
of smooth (in the norm topology) elements $\ra$, \ie 
$||\partial^\mu \ra|| < \infty$ for all multi-indices  $\mu$.\footnote{We use 
the notation $\mu = (\mu_1, \dots , \mu_n)$, \, 
$\partial^\mu_x = \partial^{\mu_1}_{x_1}  \cdots \partial^{\mu_n}_{x_n}$
and $ (\partial^\mu \ra)(x) = \partial^\mu_x \ra (x) $, where $x_k$ are the 
components of $x$ with respect to a fixed orthonormal basis in $\RR^n$ and 
$\partial_{x_k}$ are the corresponding partial derivatives, $k = 1, \dots n$.} 
Note that the elements $C \in \BH$ act as multipliers on $\IC$,
if one identifies $C$ with the corresponding constant function in $\IC$ 
(denoted by the same symbol). Clearly, the maps $(C, \ra) \mapsto C \ra$ 
and $(C, \ra) \mapsto \ra C$ are norm continuous in both variables,
$|| C \ra || \leq || C || \, || \ra || \geq || \ra C ||$, and the 
multiplication by $C$ commutes with the operations of
differentiation and integration on $\IC$. 

     In the subsequent analysis we shall find it necessary to integrate the 
functions $x \mapsto \ra(x)$. In order to handle the fact that these 
functions are in general not absolutely integrable with respect to Lebesgue 
measure due to a lack of suitable decay properties, we introduce mollifiers
$L_n : \RR^n \rightarrow \CC$. A convenient choice is given by
$$ L_n(x) = (i + x_1 + \dots + x_n)^{-1} \prod_{k = 1,  \dots ,n} 
(i + x_k)^{-1} \, , \quad x \in \RR^n \, . $$
Because of this simple form one easily verifies that
$(\partial^\mu L_n)(x) = N_{n, \mu}(x) L_n(x)$, where $N_{n, \mu}$ is smooth and 
bounded for any multi-index $\mu$; moreover $L_n \in L^1(\RR^n)$. We 
therefore choose $L_n$ as a universal mollifier on $\IC$. 

     It follows from the preceding remarks that, for any multi-index $\mu$, 
the functions $x \mapsto \partial^\mu_x \big( L_n(x) \ra(x) \big)$
are Bochner integrable in $\BH$ with respect to the Lebesgue measure. 
Moreover, applying Leibniz's  rule, one gets 
$$ \int \! dx \, || \partial^\mu_x \big( L_n(x) \ra(x) \big) || 
\leq c_{n , \mu} \, || \ra ||_{|\mu|} \, , \quad \ra \in \IC \, , $$
where $c_{n , \mu}$ does not depend on $\ra$, and we have introduced the
norms, $m \in \NN_0$,  
$$  || \ra ||_m  = \sum_{\mu, \, |\mu| \leq m} || \partial^\mu
\ra || \, . $$

     The following technical lemma is a basic ingredient in the subsequent 
discussion. In its proof, we make use of arguments furnished by Rieffel 
\cite{Ri}. 

\begin{lemma} \label{1.1}
Let $\ra, \rb \in \RC$ be $n+1$ times continuously differentiable and 
let $f \in {\cal S} (\RR^{n} \times \RR^{n})$ with $f(0,0) = 1 $. \\[1mm] 
(i) The norm limit of Bochner integrals in $\BH$, 
$$ \lim_{\varepsilon \rightarrow 0} \ (2 \pi)^{-n} \! 
\iint  \! dx dy \,  f(\varepsilon x, \varepsilon y)  \, e^{- i xy}  \, 
\ra(x) \rb(y) \doteq \ra \rtm \rb \, , $$
exists and does not depend on $f$. {Here $x y, \, x,y \in \RR^n $ is any 
symmetric bilinear form on $\RR^n$} with 
determinant $1$ or $-1$. \\[1mm]
(ii) \ With $L_n$ as above, there exists a polynomial $u,v \mapsto P_n(u,v)$ 
on $\RR^n \times \RR^n$ of degree $n + 1$ in the components of 
$u$ and $v$, respectively, such that
$$ \ra \rtm \rb = (2 \pi)^{-n} \! 
\iint \! dx dy  \, e^{- i xy}  \, P_n(\partial_x, \partial_y) \,  
L_n(x) \ra(x) \, L_n(y) \rb(y) \, , $$
where the integral is defined as a Bochner integral in $\BH$. \\[1mm]
(iii) \, $ ||  \ra \rtm \rb || \leq  c_n \, || \ra ||_{n+ 1} \,
|| \rb ||_{n+1} $,  \
for a universal constant $c_n$.  \\[1mm]
(iv) Let $C \in \BH$. Then
\begin{equation*} 
\begin{split}
& (C \ra \rtm \rb) = C (\ra \rtm \rb), \quad 
(\ra \rtm \rb C) = (\ra \rtm \rb ) C, \\
&  (\ra C \rtm \rb) = (\ra \rtm C \rb) \, , 
\end{split}
\end{equation*}
and the linear map
$$ C \mapsto \ra \rtm C \rb $$
is continuous on the unit sphere of $\BH$ in the strong operator topology.
\end{lemma}

\begin{proof} (i) Crucial for the result are certain properties of 
the function $x,y \mapsto e^{- i xy}$. Namely, for each polynomial 
$x,y \mapsto Q_n(x,y)$ of degree $n+1$ in the variables $x$ and $y$, 
respectively, there is a corresponding polynomial $P_n$
in the same family such that
$$Q_n(x,y) \,  e^{ - i xy} = P_n( - \partial_x, - \partial_y) \,
e^{- i xy} \, ,
$$
and \textit{vice versa}. The step from the right hand side to the left hand 
side is easily accomplished by differentiation; the opposite direction can 
likewise be established, noticing that the Fourier transform 
of $x,y \mapsto e^{ - i xy}$ is again of this form. Choosing 
\begin{align*}
Q_n(x, y) = L_n(x)^{-1} \,  L_n(y)^{-1}, \quad x,y \in \RR^n ,
\end{align*}

with $L_n$ specified above, one observes that 
\begin{equation*}
\begin{split}
& \iint \! dx dy  \,   f(\varepsilon x, \varepsilon y) \, e^{- i xy}  
\, \ra(x) \rb(y)  \\
& = \iint \! dx dy  \, Q_n(x,y) \,e^{- i xy} \, f(\varepsilon x, \varepsilon y)  
\, L_n(x) \ra(x) \, L_n(y) \rb(y)  \\
& = \iint \! dx dy  \, \big( P_n(-\partial_x, - \partial_y)) \,e^{- i xy}
\big) \, 
f(\varepsilon x, \varepsilon y)  \, L_n(x) \ra(x) \,
L_n(y) \rb(y) \\
& = \iint \! dx dy  \, e^{- i xy} \,  P_n(\partial_x, \partial_y)
f(\varepsilon x, \varepsilon y)  \, L_n(x) \ra(x) \,
L_n(y) \rb(y) \, ,
\end{split}
\end{equation*}
where $P_n$ is the polynomial corresponding to the chosen $Q_n$. In view of 
the smoothness and rapid decay properties of the integrands, the integrals 
are defined as Bochner integrals in $\BH$, and the last equality is obtained 
by partial integration. Decomposing $ P_n(\partial_x, \partial_y)$ into a sum 
of monomials of the form $\partial_x^\mu \partial_y^\nu$ with 
$|\mu|, |\nu| \leq n +1 $, performing the differentiations and taking into 
account the properties of $L_n$, one gets by an application of the
dominated convergence theorem 
\begin{equation} \label{eq1.1}
\begin{split}
& \lim_{\varepsilon \rightarrow 0} 
\iint \! dx dy  \, e^{-i xy} \,  P_n(\partial_x, \partial_y)
f(\varepsilon x, \varepsilon y)  \, L_n(x) \, \ra(x) \,
L_n(y) \rb(y) \\
& = \iint \! dx dy  \, e^{-i xy}  \, P_n(\partial_x, \partial_y) \, 
L_n(x) \ra(x) \, L_n(y) \rb(y) \, .
\end{split} 
\end{equation}
Note that the derivatives of $x,y \mapsto f(\varepsilon x, \varepsilon y)$ 
contain powers of $\varepsilon$ as factors and therefore disappear in the 
limit. In particular, the limit does not depend on the choice of $f$. 

     Assertion (ii) has been established in (\ref{eq1.1}). From this relation 
one also obtains the estimate (iii), because of the properties of $L_n$ 
established above. 

     The proof of the equalities in (iv) is another straightforward 
consequence of (ii) and is therefore omitted. It remains to establish
the continuity of the map. In view of the continuity and decay 
properties of the functions appearing in the representation  
(ii) of the product $\rtm$, the integrals underlying the definition of  
$\ra \rtm C \rb$ can be approximated in norm, uniformly for $C \in \BH$, 
$|| C || \leq 1 $, by finite sums of the form
$$ \sum_{\mu, \nu, i, k} c_{\mu, \nu, i, k} \, 
(\partial^\mu \ra)(x_i) 
\, C \, (\partial^\nu \rb)(y_k)
\, , $$
where $ c_{\mu. \nu, i, k}$ are constants which do not depend on
$C$. Since the operators $ (\partial^\mu \ra)(x_i)$,
$(\partial^\nu \rb)(y_k)$ are bounded, the stated
continuity properties of the map with respect to $C$ then follow.
\end{proof}

     Within this general setting, the Rieffel deformations of the 
$C^*$--dynamical system $(\Cs, \RR^n)$ \cite{Ri} can be presented as follows. 
Let $\Cs^\infty \subset \Cs$ be the ${}^*$--algebra of smooth elements with 
respect to the action of $\alpha$ and let $Q$ be a real skew symmetric matrix 
relative to the chosen bilinear form on $\RR^n$, \ie $x Q y = -  y Q x$, 
$x,y \in \RR^n$. One then considers  for $A, B \in \Cs^\infty$ 
the functions in $\IC$ given by
$$ x \mapsto \boldsymbol{A_{\alpha_Q}} (x) \doteq  \alpha_{Qx}(A) \, , \quad  
y \mapsto \boldsymbol{B_{\alpha}} (y) \doteq  \alpha_{y}(B) $$
and sets  
$$ A \tq B \doteq  \boldsymbol{A_{\alpha_Q}} \rtm  \boldsymbol{B_{\alpha}} 
\, , \quad  A, B \in \C^\infty \, . $$
It has been shown by Rieffel  \cite{Ri} that $\tq$ defines an associative
product on $\Cs^\infty$, the Rieffel product, which is compatible with the 
${}^*$--operation. In view of the normalization of the bilinear form 
chosen in Lemma \ref{1.1} (i), the original identity operator $1$ still
acts as the identity with respect to 
the new product $\tq$. Moreover, there exists a 
$C^*$--norm on the deformed algebra $(\Cs^\infty, \tq)$. It is of interest here 
that, by Lemma \ref{1.1}, the Rieffel product extends to more general 
functions in $\RC$ in a natural manner. We shall take advantage of this fact 
in the following subsection. 

\subsection{Warped Convolution}

     We turn now to the discussion of the warped convolution on $(\Cs, \RR^n)$ 
introduced in \cite{BuSu2}.  Many of the results below
were stated there and provided with sketches of proofs; in addition
to supplying complete proofs of those assertions, here we also prove results 
which strengthen and complement those discussed in \cite{BuSu2}. 

     The weakly continuous unitary representation $U$ of $\RR^n$ implementing 
$\alpha$ enters into the definition of the warped convolution. 
In a first step we want to give proper meaning to the 
formal expressions $\int_B \, \alpha_{Qx}(A) \, dE(x)$ and 
$\int_B \, dE(x) \, \alpha_{Qx}(A) $, $A \in \Cs^\infty$, where $Q$ is any real 
$n \times n$ matrix, $E$ is the spectral resolution of $U$ and 
$B \subset \RR^n$ is any bounded Borel set. If $F$ is a  
finite--dimensional projection, it follows from the spectral
calculus that the integrals $\int_B \, \alpha_{Qx}(A) \, F dE(x)$ and 
$\int_B \, dE(x) F \, \alpha_{Qx}(A) $ are well--defined in the
strong operator topology. Moreover, since 
$U(y) = \int \! e^{i xy} \, dE(x)$, $y \in \RR^n$, one obtains 
for any test function $f$ as in Lemma \ref{1.1}
\begin{equation*} 
\begin{split}
& \int_B  \alpha_{Qx}(A) \, F dE(x) =  \lim_{\varepsilon \rightarrow
  0} \ (2 \pi)^{-n} 
\iint \! dx dy \, f(\varepsilon x, \varepsilon y) 
e^{- ixy} \, \alpha_{Qx}(A) \, F E(B) U(y) \, , \\
& \int_B  dE(x) \, F \alpha_{Qx}(A)  =  \lim_{\varepsilon \rightarrow
  0} \ (2 \pi)^{-n} 
\iint \! dx dy \, f(\varepsilon x, \varepsilon y) 
e^{- ixy} \ E(B) U(x) F \, \alpha_{Qy}(A) \, .
\end{split}
\end{equation*}
Hence, adopting the notation in Lemma \ref{1.1},  
$$ \int_B  \alpha_{Qx}(A) \, F dE(x)
= \boldsymbol{A_{\alpha_Q}} \rtm F \boldsymbol{U_B} \, , \quad
 \int_B  dE(x) \, F \alpha_{Qx}(A) 
= \boldsymbol{U_B} F \rtm  \boldsymbol{A_{\alpha_Q}} \, , $$ 
where we have introduced the function 
$x \mapsto \boldsymbol{U_B}(x) \doteq E(B) \, U(x)$, 
which is an element of $\IC$ since the set $B$ is bounded. Choosing any net 
of projections $F$ converging monotonically to the identity operator
$1$, it follows from part~(iv) of Lemma \ref{1.1} that 
the right hand side of these equalities converges in the strong operator 
topology. Hence the corresponding limits of the integrals on the left hand side 
exist in $\BH$ and can be used to define the warped convolution 
integrals as
\begin{equation}  \label{eq2.1}
\begin{split}
& \int_B  \alpha_{Qx}(A) \, dE(x)
\doteq
\lim_{F\nearrow1} \int_B  \alpha_{Qx}(A) \,F\, dE(x)
=
\boldsymbol{A_{\alpha_Q}} \rtm \boldsymbol{U_B} \, , \\
&  \int_B  dE(x) \, \alpha_{Qx}(A)
\doteq
\lim_{F\nearrow1} \int_B  dE(x)\,F\alpha_{Qx}(A)
= 
\boldsymbol{U_B} \rtm  \boldsymbol{A_{\alpha_Q}} \, .
\end{split}
\end{equation}
If the spectrum $\mbox{sp} \, U$ of $U$ is compact, the integrals do not 
depend on $B$ if $B \supset \mbox{sp} \, U$, so one can proceed to the 
limit $B \nearrow \RR^n$ in $\BH$. In the general case, however, we have 
\textit{a~priori} no control on the continuity properties of the resulting 
operators. In order to cope with this problem, we consider the dense domain 
${\cal D \subset H}$ of vectors which are smooth with respect to the action 
of $U$. Let $P = (P_1, \dots ,P_n)$ be the generators of $U$ and let
$$x \mapsto \big( \boldsymbol{U_{\RR^n}} (1 + P^2)^{-n-1} \big)(x) \doteq 
U(x)  (1 + P^2)^{-n-1} \, . $$ 
This function is an element of $\RC$ which is $n + 1$ times continuously 
differentiable. Making use of the first half of part (iv) of Lemma \ref{1.1}, 
we thus get for $\Phi \in {\cal D}$
\begin{equation*}
\begin{split}
& \big( \boldsymbol{A_{\alpha_Q}} \rtm \boldsymbol{U_B} \big) \, \Phi = 
\big( \boldsymbol{A_{\alpha_Q}} \rtm 
\boldsymbol{U_B} \big) (1 + P^2)^{-n-1} \, (1 + P^2)^{n+1} \Phi \\
& = \big( \boldsymbol{A_{\alpha_Q}} \rtm E(B) \,
\boldsymbol{U_{\RR^n}} (1 + P^2)^{-n-1} \big) \, (1 + P^2)^{n+1} \Phi \, .
\end{split}
\end{equation*}
In the latter expression we can proceed to the limit 
$B \nearrow \RR^n$ according to the second half of part (iv) of Lemma 
\ref{1.1}, since the projections $E(B)$ converge strongly to $1$ in this 
limit. In view of relation (\ref{eq2.1}), this proves the existence of the 
first type of integrals, 
\begin{equation} \label{covariance}
\begin{split}
& \int  \alpha_{Qx}(A) \, dE(x) \, \Phi \doteq
\lim_{B \nearrow \RR^n} \int_B  \alpha_{Qx}(A) \, dE(x) \, \Phi \\
& = \big( \boldsymbol{A_{\alpha_Q}} \rtm \boldsymbol{U_{\RR^n}} (1 + P^2)^{-n-1} \big) \, 
(1 + P^2)^{n+1} \Phi \, .
\end{split}
\end{equation}
Part (iii) and the first half of part (iv) of Lemma \ref{1.1} 
imply that the functions 
$x \mapsto \alpha_x 
\big( \boldsymbol{A_{\alpha_Q}} \rtm \boldsymbol{U_{\RR^n}} (1 +
P^2)^{-n-1} \big)
= U(x) \boldsymbol{A_{\alpha_Q}} U(x)^{-1} \rtm \boldsymbol{U_{\RR^n}} (1 +
P^2)^{-n-1}$
are elements of $\IC$, since $x,y \mapsto (U(x)
\boldsymbol{A_{\alpha_Q}} U(x)^{-1})(y) = \alpha_{x + Qy}(A)$ is smooth
in both variables. So it is also clear that the domain $\cal D$ 
is stable under the action of the integrals. 

     For the proof of existence of the second type of integrals, we 
make use of the fact that the functions 
$x \mapsto \alpha_x \big( (1 + P^2)^{n+1} A (1 + P^2)^{-n-1} \big)$, 
$A \in \Cs^\infty$ are elements of $\IC$.\footnote{This can be seen by pulling 
the components of $P$ through from the left to the right of $A$, noticing that 
their commutators with $A$ yield derivatives of $A$.}  Thus, by the first 
half of part (iv) of Lemma \ref{1.1}, we get for $\Phi \in {\cal D}$
\begin{equation*}
\begin{split}
& \big( \boldsymbol{U_B} \rtm \boldsymbol{A_{\alpha_Q}} \big) \, \Phi 
 =  \big( \boldsymbol{U_B} \rtm \boldsymbol{A_{\alpha_Q}} \big)
\, (1 + P^2)^{-n-1}  (1 + P^2)^{n+1} \Phi \\
& =  \big( \boldsymbol{U_B} \rtm (1 + P^2)^{-n-1} (1 + P^2)^{n+1}  \boldsymbol{A_{\alpha_Q}} 
(1 + P^2)^{-n-1} \big) \,  (1 + P^2)^{n+1} \Phi \\
& =  \big( \boldsymbol{U_{\RR^n}} E(B) (1 + P^2)^{-n-1}   
\rtm (1 + P^2)^{n+1}  \boldsymbol{A_{\alpha_Q}} 
(1 + P^2)^{-n-1} \big)  (1 + P^2)^{n+1} \Phi \, ,
\end{split}
\end{equation*}
where in the latter expression we can proceed again to the 
strong limit $E(B) \nearrow 1$ if $B \nearrow \RR^n$. 
In view of relation (\ref{eq2.1}), this proves existence of the strong limits
\begin{equation*}
\begin{split}
& \int  dE(x) \, \alpha_{Qx}(A) \, \Phi \doteq
\lim_{B \nearrow \RR^n} \int_B  dE(x) \, \alpha_{Qx}(A) \, \Phi \\
& =  \big( \boldsymbol{U_{\RR^n}} (1 + P^2)^{-n-1} 
\rtm (1 + P^2)^{n+1}  \boldsymbol{A_{\alpha_Q}} 
(1 + P^2)^{-n-1} \big)  (1 + P^2)^{n+1} \Phi \, .
\end{split}
\end{equation*}
By a similar argument as before, one finds that the domain $\cal D$ is stable 
under the action of the second type of integrals, as well.

     It follows at once from the preceding results and Lemma \ref{1.1} that 
one can conveniently present the two types of integrals on the domain 
${\cal D}$ in terms of the strong limits   
\begin{equation} \label{convenient} 
\begin{split}
& \int  \alpha_{Qx}(A) \, dE(x) \, \Phi 
=  \lim_{\varepsilon \rightarrow 0} \ (2 \pi)^{-n} \! \iint \! dx dy
\, f(\varepsilon x, \varepsilon y) \, e^{-ixy} \, 
\alpha_{Qx}(A) \, U(y) \, \Phi \, , \\
& \int  dE(x) \, \alpha_{Qx}(A) \, \Phi 
=  \lim_{\varepsilon \rightarrow 0} \ (2 \pi)^{-n} \! \iint \! dx dy
\,   f(\varepsilon x, \varepsilon y) \,  e^{-ixy} \, 
U(x) \, \alpha_{Qy}(A)  \, \Phi \, . 
\end{split}
\end{equation}
We shall make use of these relations throughout the subsequent 
analysis. 

     In the following, we limit ourselves to the particularly interesting case 
where the matrix $Q$ entering into the definition of the integrals is skew 
symmetric relative to the chosen bilinear form on $\RR^n$. 
It is understood without further mention that the integrals 
are defined on the common stable domain $\cal D$. 

\begin{lemma} \label{symmetric}
Let $Q$ be any real skew symmetric matrix on $\RR^n$ and
let $A \in \Cs^\infty$.  Then \\[1mm]
(i) \ $\int \! dE(x) \, \alpha_{Qx}(A) = \int \! \alpha_{Qx}(A) \,
dE(x)$ and \\[1mm]
(ii) \ $\big( \int \! \alpha_{Qx}(A)  \, dE(x) \big)^* \supset
\int \! \alpha_{Qx}(A^*) \, dE(x) $.  
\end{lemma}
 
\begin{proof} (i) As pointed out above, one has for $\Phi \in {\cal D}$
$$ \int  \alpha_{Qx}(A) \, dE(x) \, \Phi  = 
\lim_{\varepsilon \rightarrow 0} \ \, (2 \pi)^{-n} \! 
\iint \! dx dy \, f(\varepsilon x, \varepsilon y) \, 
e^{-ixy} \, \alpha_{Qx}(A) \, U(y) \, \Phi \, . $$
The integration can be restricted to the submanifold 
$(\mbox{ker} \, Q)^\perp \times (\mbox{ker} \, Q)^\perp$, since  
the remaining integrals merely produce factors of $2 \pi$. 
Substituting $x \rightarrow x + Q^{-1}y$ and taking into
account that $Q^{-1}$ is skew symmetric, one gets 
\begin{equation*}
\begin{split}
& \iint \! dx dy \, f(\varepsilon x, \varepsilon y) \, 
e^{-ixy} \, \alpha_{Qx}(A) \, U(y)   \\
& = \iint \! dx dy \, g(\varepsilon x, \varepsilon y) \, 
e^{-ixy} \, U(y) \alpha_{Qx}(A)  \\
& =  \iint \! dx dy \, g(\varepsilon y, \varepsilon x) \, 
e^{-ixy}  \, U(x) \, 
\alpha_{Qy}(A) \, ,
\end{split}
\end{equation*}
where $g(x,y) \doteq f(x + Q^{-1}y, y)$. In the limit 
$\varepsilon \rightarrow 0$ one obtains 
$$  \lim_{\varepsilon \rightarrow 0} \ (2 \pi)^{-n} \!
\iint \! dx dy \, g(\varepsilon y, \varepsilon x) \, 
e^{-ixy}  \, U(x) \, 
\alpha_{Qy}(A) \, \Phi = \int  dE(x) \, \alpha_{Qx}(A) \, \Phi  
\, ,$$
proving assertion (i). \\[1mm]

(ii) For the proof of the second assertion, note that 
\begin{equation*}
\begin{split}
&  \Big( \iint \! dx dy \, f(\varepsilon x, \varepsilon y) \, 
e^{-ixy} \, \alpha_{Qx}(A) \, U(y) \Big)^*  \\
& =  \iint \! dx dy \, \overline{f(\varepsilon x, \varepsilon y)} \, 
e^{ixy} \,  U(-y) \, \alpha_{Qx}(A^*) \\
& =  \iint \! dx dy \, \overline{f(\varepsilon y, - \varepsilon x)} \, 
e^{-ixy} \,  U(x) \, \alpha_{Qy}(A^*)  \, .
\end{split}
\end{equation*}
Hence, for $\Psi, \Phi \in {\cal D}$, 
\begin{equation*}
\begin{split}
\langle \Psi, \int  \! \alpha_{Qx}(A) \, dE(x) \, \Phi \rangle 
& = \lim_{\varepsilon \rightarrow 0} \, (2\pi)^{-n}
\langle \Psi, \iint \! dx dy \, f(\varepsilon x, \varepsilon y) \, 
e^{-ixy} \, \alpha_{Qx}(A) \, U(y) \Phi \rangle \\
& =  \lim_{\varepsilon \rightarrow 0} \,   (2\pi)^{-n}
\langle \iint \! dx dy \, \overline{f(\varepsilon y, - \varepsilon x)} \, 
e^{-ixy} \,  U(x) \, \alpha_{Qy}(A^*) \Psi, \Phi \rangle \\
& = \langle \int  \!  dE(x) \ \alpha_{Qx}(A^*) \, \Psi,  \Phi \rangle \, . 
 \end{split}
\end{equation*}
The assertion now follows from the preceding step. 
\end{proof}

     We may therefore meaningfully declare the following definition
as in \cite{BuSu2}.

\begin{definition} Let $Q$ be a real skew symmetric matrix on
$\RR^n$ and let $A \in \Cs^\infty$. The corresponding warped convolution $A_Q$ 
of $A$ is defined on the domain ${\cal D}$ by means of the preceding results 
according to   
$$ A_Q \doteq \int \! dE(x) \, \alpha_{Qx}(A) = \int \! 
\alpha_{Qx}(A) \, dE(x) \, .$$
In particular, $1_Q = 1$.
\end{definition}

     We shall next show that the warped convolution provides a
representation of the algebra $(\Cs^\infty, \tq)$  defined by
$A \mapsto \pi_Q(A) \doteq A_Q$. The argument proceeds through a 
number of steps. It is apparent from the definition that the map
$\pi_Q$ is linear and, by the second part of the preceding lemma, we have
$\pi_Q(A)^* \supset \pi_Q(A^*)$. The proof that $\pi_Q$  is also 
multiplicative requires more work.

\begin{lemma} \label{rep}
Let $Q$ be a real skew symmetric matrix on $\RR^n$ and let $A, B \in \Cs^\infty$.
Then (understood as an equality on $\Ds$)
$$ A_Q  B_Q = (A \tq B)_Q \, , $$
where $\tq$ denotes the Rieffel product on $\Cs^\infty$. In other words, 
$$\pi_Q(A) \pi_Q(B) = \pi_Q(A \tq B) \, .$$ 
\end{lemma}

\begin{proof} Let $f,g$ be test functions as in Lemma \ref{1.1}.
Recalling that $A \tq B \in \Cs^\infty$ can be approximated in norm according to 
$$ A \tq B =
\lim_{\delta \rightarrow 0} \ (2 \pi)^{-n} \! \iint \! dv dw \, 
f(\delta v, \delta w) \,  e^{-ivw} \,  
\alpha_{Qv}(A) \alpha_{w}(B)  \, , $$
one finds for $\Phi \in {\cal D}$
\begin{equation*} \scriptscriptstyle  
\begin{split} 
&  (2 \pi)^{2n}  (A \tq B)_Q \Phi \\
= & \lim_{\varepsilon, \delta \rightarrow 0} \, 
\iiiint \! dv dw dx dy \, 
 f(\delta v, \delta w) g(\varepsilon x, \varepsilon y) \,  e^{-ivw - ixy} \,  
\alpha_{Qv + Qx}(A) \alpha_{w + Qx}(B)  \, U(y) \, \Phi \, ,
\end{split}
\end{equation*}
in the sense of strong convergence. Similarly, one has 
\begin{equation*} 
\begin{split} 
& (2 \pi)^{2n}  A_Q B_Q \Phi \\
 =  &  \lim_{\varepsilon, \delta \rightarrow 0} \, 
\iiiint \! dv dw dx dy \, g(\varepsilon v, \varepsilon w) f(\delta x,
\delta y) \, e^{-ivw -ixy}  \alpha_{Qv}(A) \, U(w) \, 
\alpha_{Qx}(B) \, U(y) \Phi \, .
\end{split}
\end{equation*}
In both cases the limits are to be performed in the given order. In order to 
see that these limits coincide, one rewrites the two integrals. For the first 
one, by substituting $(v,w) \rightarrow (v - x, w - Qx) $ and bearing in 
mind that $Q$ is skew symmetric, one obtains
\begin{equation*} 
\begin{split} 
& \iiiint \! dv dw dx dy \, 
 f(\delta v, \delta w) g(\varepsilon x, \varepsilon y) \,  e^{-ivw - ixy} \,  
\alpha_{Qv + Qx}(A) \alpha_{w + Qx}(B)  \, U(y) \, \Phi \\
& = \iiiint \! dv dw dx dy \,  
h_{\delta, \varepsilon}(v,w,x,y) \, e^{-ivw - ix(y +Qv - w)} \, 
\alpha_{Qv}(A) \, \alpha_{w}(B) \, U(y) \, \Phi
\, , 
\end{split}
\end{equation*}
where $h_{\delta, \varepsilon}(v,w,x,y) \doteq 
f(\delta (v - x), \delta (w - Qx)) \, g(\varepsilon x, \varepsilon y) \, $. 
For the second integral, by substituting 
$(w,y) \rightarrow (w - Qx, y +Qx - w) $ and
making use again of the fact that $Q$ is skew symmetric, one finds 
\begin{equation*} 
\begin{split} 
& \iiiint \! dv dw dx dy \, g(\varepsilon v, \varepsilon w) f(\delta x,
\delta y) \, e^{-ivw -ixy}  \alpha_{Qv}(A) \, U(w) \, 
\alpha_{Qx}(B) \, U(y) \Phi \\
& = \iiiint \! dv dw dx dy \,  k_{\delta, \varepsilon}(v,w,x,y) \, 
e^{-ivw - ix(y +Qv -w)} \,
\alpha_{Qv}(A) \, \alpha_{w}(B) \, U(y) \, \Phi \, ,
\end{split}
\end{equation*}
where $k_{\delta, \varepsilon}(v,w,x,y) =  f(\delta x, \delta (y + Qx -w)) 
g(\varepsilon v, \varepsilon (w - Qx))$. Thus the two integrals
coincide apart from the mollifying test functions 
$h_{\delta, \varepsilon}$ and $k_{\delta, \varepsilon}$, respectively. 

     In order to show that the two integrals have the same limits, one proceeds
as in the proof of part (i) of Lemma \ref{1.1}. Again one finds by 
Fourier transformation that for any given polynomial $L$ on $\RR^{4n}$ there 
is a corresponding polynomial $P$ such that
$$  L(v,w,x,y) \, e^{-ivw - ix(y +Qv -w)}
= P(-\partial_v, -\partial_w, -\partial_x, -\partial_y)  \, e^{-ivw - ix(y
  +Qv -w)} \, . $$
A convenient choice for $L$ is given by
\begin{align}\label{polynomial-L}
 L(v,w,x,y) = \big( L_n(v)  L_n(w)  L_n(x)  L_n(y) \big)^{-1} \, ,
\end{align}
where $L_n$ are the mollifiers introduced before. With this choice one gets 
by partial integration, setting $\boldsymbol{u} \doteq (v,w,x,y) $,
$d\boldsymbol{u} = dvdwdxdy$ and  
$\boldsymbol{\partial} \doteq 
(\partial_v, \partial_w, \partial_x, \partial_y) $, 
\begin{equation*} 
\begin{split} 
& \iiiint \! d\boldsymbol{u} \,  h_{\delta, \varepsilon}(\boldsymbol{u}) \, 
e^{-ivw - ix(y +Qv -w)} \,
\alpha_{Qv}(A) \, \alpha_{w}(B) \, U(y) \, \Phi \\
= & \iiiint \! d\boldsymbol{u} \, e^{-ivw - ix(y +Qv -w)} \, 
P(\boldsymbol{\partial}) \, 
h_{\delta, \varepsilon}(\boldsymbol{u}) \, L(\boldsymbol{u})^{-1}
\alpha_{Qv}(A) \, \alpha_{w}(B) \, U(y) \, \Phi \, .
\end{split}
\end{equation*}
The derivatives in the second line are well--defined, since
$A,B \in \Cs^\infty$ and \mbox{$\Phi \in {\cal D}$}. Moreover, all derivatives
of $\boldsymbol{u} \mapsto  L(\boldsymbol{u})^{-1} $ are absolutely 
integrable, and the derivatives of 
$\boldsymbol{u} \mapsto h_{\delta, \varepsilon}(\boldsymbol{u})$ produce factors of 
$\delta$ and $\varepsilon$, respectively. Thus in the limit of small 
$\delta$ and $\varepsilon$ one can replace in the above integral the 
test function $h_{\delta, \varepsilon}$ by its value at the origin, \ie $1$. The 
same argument applies if one replaces
$h_{\delta, \varepsilon}$  by $k_{\delta, \varepsilon}$, proving
equality of the limits of the respective integrals.
\end{proof}

     At this point, the operators $\pi_Q(A) = A_Q$, $A \in \Cs^\infty$, are 
well--defined only on the dense, invariant domain $\Ds$, defining 
there a ${}^*$--algebra. We next show that they may be extended to bounded 
operators on $\Hs$, in contradiction to an assertion made in \cite{BuSu2}. 
In the proof we make use of the fact that the algebra $(\Cs^\infty,\tq)$ admits 
a $C^*$--norm $|| \cdot ||_Q$, \cf \ \cite[Ch.\ 4]{Ri}. It thus can be 
completed to a $C^*$--algebra, denoted by $(\Cs_Q,\tq)$, to which the group of
automorphisms $\alpha_x$, $x \in \RR^n$, extends in a strongly continuous manner
\cite[Prop.\ 5.11]{Ri}. 

     As is well known, every positive element of a $C^*$--algebra has a 
positive square root in the algebra. We need here the following more 
detailed information. 

\begin{lemma}  \label{positive}
Let $A \in \Cs^\infty$ be strictly positive in $(\Cs_Q, \tq)$, \ie
$A - \delta \, 1 = \! B^* \tq B$ for some $\delta > 0$
and $B \in  (\Cs_Q, \tq)$. Then its positive square root 
$\sqrt{A} \in  (\Cs_Q, \tq)$ is also an element of $\Cs^\infty$.
\end{lemma}
\begin{proof}
The form of $A$ implies that its spectrum is contained in the interval 
$[\delta, \, \|A\|_Q]$. As the square root $z \mapsto \sqrt{z}$ is holomorphic 
in a complex neighborhood of this region and $\Cs^\infty$ is closed
under the holomorphic calculus
\cite[Corollary 7.6]{Ri}, it follows that $\sqrt{A} \in \Cs^\infty$.
\end{proof}

     With the help of this lemma we can show now that the operators
$\pi_Q(A) = A_Q$, $A \in \Cs^\infty$ are bounded. For the 
operators $(\delta + a)^2 1 - A^* \tq A$, $a \doteq \| A \|_Q$, are elements
of $\Cs^\infty$ and strictly positive in $(\Cs_Q, \tq)$
for every $\delta > 0$. Thus their positive square roots 
$B \doteq \sqrt{(\delta + a)^2 1 - A^* \tq A} \in (\Cs_Q, \tq)$
are elements of $\Cs^\infty$ according to the preceding lemma. 
Bearing in mind the properties of $\pi_Q$ established thus far, we
therefore have for any $\Phi \in {\cal D}$
\begin{equation*}
\begin{split}
& (\delta + a)^2 \| \Phi \|^2 - \| \pi_Q (A) \Phi \|^2 \\
& = \langle \Phi, \pi_Q \big((\delta + a)^2 1 - A^* \tq A \big) 
\Phi \rangle 
= \langle \Phi, \pi_Q (B^* \tq B) \Phi \rangle = \| \pi_Q(B) \Phi \|^2 \geq 0, 
\end{split}
\end{equation*}
where we made use of the fact that $B^* = B$ since $B$ is positive.
Hence we obtain 
$\| \pi_Q(A) \Phi \| \leq (\| A \|_Q + \delta) \, \| \Phi \|$, $\Phi \in \Ds$. 
Since $\Ds$ is dense in $\Hs$ and $\delta > 0$ was arbitrary, we conclude that 
the operators $\pi_Q(A)$ can be extended to the whole Hilbert space
with operator norms satisfying the bound 
\begin{equation} \label{bound}
\| \pi_Q (A) \| \leq \| A \|_Q \, , \quad A \in \Cs^\infty \, . 
\end{equation}
Moreover, it follows from this estimate and the preceding results that the 
representation $\pi_Q : \Cs^\infty \rightarrow \Bs(\Hs)$ can be continuously 
extended to a representation of the $C^*$--algebra  $(\Cs_Q, \tq)$ on $\Hs$.
We summarize these findings.

\begin{theorem} \label{warpedrep}
The map
\begin{align*}
 \pi_Q(A) \doteq A_Q \, , \qquad A \in \Cs^\infty\,,
\end{align*}
extends to a representation of the Rieffel--deformed $C^*$--algebra 
$(\Cs_Q, \tq)$ on $\Hs$. In particular, one has the bound
$$ \| \pi_Q(A) \| \leq \| A \|_Q \, ,  \quad A \in  (\Cs_Q, \tq) \, . $$
\end{theorem}
\noindent \textbf{Example:} Of particular interest in physics are the cases 
where the spectrum of $U$ contains an atomic part. Without loss of generality 
one may then assume that $\{ 0 \}$ is part of the atomic spectrum with 
corresponding invariant vector~$\Omega$.\footnote{Note that proceeding from 
the group $U(x)$ to the group $U_q(x) = e^{iqx} U(x)$, $x \in \RR^n$,
merely amounts to a translation $A \rightarrow \alpha_{Qq}(A)$ of the operators 
$A$ in the original warped convolution.} 
Since the algebra $\Cs^\infty$ is weakly dense in $\BH$, it is clear 
that $\Omega$ is cyclic for  $\Cs^\infty$; moreover, because of the invariance
of $\Omega$ under the action of $U$, one also has  
$\Cs^\infty \Omega \subset {\cal D}$. Within this setting the
relation between the warped convolutions and the Rieffel deformations
can be exhibited quite easily. For, as a 
consequence of the invariance of $\Omega$, one obtains for $A,B \in \Cs^\infty$, 
\begin{equation*}
\begin{split}
A_Q \, B \Omega  &  =  
\lim_{\varepsilon \rightarrow 0} \, (2\pi)^{-n} 
\iint \! dx dy \, e^{-ixy} \, f(\varepsilon x, \varepsilon y) \, 
\alpha_{Qx}(A) \, U(y) \,  B \Omega \\
& = \lim_{\varepsilon \rightarrow 0} \, (2\pi)^{-n}
\iint \! dx dy \, e^{-ixy} \, f(\varepsilon x, \varepsilon y) \, 
\alpha_{Qx}(A) \, \alpha_y(B) \, \Omega \\
& = (A \tq B) \, \Omega \, .
\end{split}
\end{equation*}
In particular, 
\begin{equation}  \label{vacuum}
A_Q \Omega = A \Omega \, , \quad A \in \Cs^\infty \, .
\end{equation}
Making use of the associativity of the product $\tq$ on $\Cs^\infty$,
it is therefore clear that for $A,B,C \in \Cs^\infty$,
$$ A_Q B_Q \, C \Omega =  A_Q \, (B \tq  C) \Omega 
= (A \tq B \tq C) \Omega = (A \tq B)_Q \, C
\Omega \, . 
$$

     We return now to the discussion of the general case and exhibit
further interesting properties of the representations $\pi_Q$ introduced above.

\begin{prop} \label{survey}
Let $\pi_Q$ be the representation of the $C^*$--algebra $(\Cs_Q, \tq)$
established by the preceding theorem. \\[1mm]
\noindent (i) \ $\pi_Q$ is $\alpha$--covariant, \ie for any 
$A \in (\Cs_Q, \tq)$ 
$$ \pi_Q(\alpha_x(A)) = U(x) \pi_Q(A)  U(x)^{-1} \, , \quad x \in
\RR^n  \, .$$
\noindent (ii)  \ $\! \pi_Q$ induces a bijective map of \ $\Cs^\infty$ 
 onto itself. \\[1mm]
\noindent (iii) $\pi_Q$ is faithful, \ie 
$\| \pi_Q(A) \| = \| A \|_Q$, \ $A \in (\Cs_Q, \tq)$. \\[1mm]
\noindent (iv) $\pi_Q$ is irreducible. 
\end{prop}
\begin{proof}
(i) Let $A \in \Cs^\infty$. Since the domain $\Ds$ is stable under the action 
of the unitaries $U(x)$, relation (\ref{covariance}) and Lemma \ref{1.1}
imply $U(x) A_Q U(x)^{-1} = (\alpha_x(A))_Q$, proving the assertion for 
$A \in \Cs^\infty$. The continuity properties of $\pi_Q$ and the automorphic 
action of $\alpha_x$ on $(\Cs_Q, \tq)$ then yield assertion (i). \\[1mm]
(ii) According to \cite[Thm.\ 7.1]{Ri}, 
the smooth elements of $(\Cs_Q, \tq)$ are exactly the elements of 
$\Cs^\infty$. It therefore follows from the  
continuity of the map 
$\pi_Q$ that the functions $x \mapsto \pi_Q(\alpha_x(A))$, 
$A \in \Cs^\infty$, are smooth; hence $\pi_Q(A) = A_Q \in \Cs^\infty$ for 
$A \in \Cs^\infty$. The proof that $\pi_Q \upharpoonright\Cs^\infty $ is 
bijective requires a computation: In view of the preceding observation, 
one may apply the warping procedure with underlying matrix $-Q$ to the 
operator $A_Q$, giving $(A_Q)_{-Q}$. Now according to relation (\ref{convenient})
one has on the domain $\Ds$
\begin{equation*}
\begin{split}
& (2 \pi)^{2n} \, (A_Q)_{-Q} \\
& =  \lim_{\varepsilon, \delta \rightarrow 0} \,
\iiiint \! dv dw dx dy \, f(\varepsilon v, \varepsilon w) 
f(\delta x, \delta y) \, e^{-ivw -ixy} \, \alpha_{Qx - Qv}(A) U(y)
U(w) \, ,
\end{split}
\end{equation*}
where the limits are to be performed in the given order. 
Substituting $(v,w) \rightarrow (x-v,w-y)$, the integral can 
be transformed into 
$$ 
\iiiint \! dv dw dx dy \, f(\varepsilon (x-v), \varepsilon (w-y)) 
f(\delta x, \delta y) \,   e^{ivw -ixw - iyv} \,
\alpha_{Qv}(A) U(w) \, .
$$
As the $x,y$--integration in the latter integral involves only ordinary 
functions, it is straightforward to compute its limit for 
$\delta \rightarrow 0$, giving 
$$
 (2 \pi)^{n} \, \iint \! dv dw \, 
(1/\varepsilon)^{2n} \, \widehat{f}(w/\varepsilon, - v/\varepsilon) \, 
e^{-ivw} \, \alpha_{Qv}(A) U(w) \, ,
$$
where $\widehat{f}$ denotes the Fourier transform of $f$. It is also
apparent that the latter expression converges to 
$(2 \pi)^{2n} \, A$ as $\varepsilon \rightarrow 0$. Hence
$(A_Q)_{-Q} = A$ for $A \in \Cs^\infty$. Now if $\pi_Q(A) = A_Q = 0$,
it follows that $A = (A_Q)_{-Q} = 0$, so $\pi_Q \upharpoonright
\Cs^\infty$ is injective; similarly, interchanging the role of $Q$ and $-Q$,  
one has $\pi_Q(A_{-Q}) = (A_{-Q})_Q = A$, so 
$\pi_Q \upharpoonright \Cs^\infty$ is also surjective. \\[1mm]
(iii) Since $\pi_Q$ is $\alpha$--covariant, its kernel 
$\mbox{ker} \, \pi_Q$ is $\alpha$--invariant. Hence,  
in view of the strongly continuous action of $\alpha$
on $(\Cs^\infty, \tq)$, the space $\mbox{ker} \, \pi_Q \bigcap \Cs^\infty$
is dense in $\mbox{ker} \, \pi_Q$. But this space coincides with $\{0\}$,
since $\pi_Q \upharpoonright \Cs^\infty$ is injective according to the 
preceding result. Consequently, $\| \pi_Q (\, \cdot \,) \|$ defines a 
$C^*$--norm on $(\Cs_Q, \tq)$, which must coincide with 
$\| \cdot \|_Q$ because of the uniqueness of such norms. \\ 
(iv) The final assertion follows from the fact that 
$\pi_Q \upharpoonright \Cs^\infty$ is surjective. So its range 
contains $\Cs^\infty$, which is weakly dense in $\BH$. 
\end{proof}

    Let us turn now to the case of an abstractly given $C^*$--dynamical 
system $(\As,\RR^n)$ equipped with some strongly continuous representation
$\alpha : \RR^n \rightarrow \mbox{Aut} \, \As$. Denoting by $\As^\infty$ 
the smooth elements of $\As$, one obtains by arguments given by 
Rieffel \cite{Ri} and sketched at the end of Section 2.1 
a deformed ${}^*$--algebra $(\As^\infty, \times^Q)$
with \mbox{$C^*$--norm} $\| \, \cdot \, \|^Q$ for any given skew 
symmetric matrix $Q$.
Its $C^*$--completion will be denoted $(\As^Q, \tQ)$. Here we have
used $Q$ as an upper index in order to distinguish the abstract
setting from the concrete one used thus far. 

     Let $(\pi, \Hs)$ be an $\alpha$--covariant 
representation of $\As$ on a Hilbert space $\Hs$, \ie on $\Hs$ there exists 
a weakly continuous unitary representation $U$ of $\RR^n$ such that
\begin{align}  \label{pi-covariance}
 U(x)\pi(A)U(x)^{-1} = \pi(\alpha_x(A)) \, , \qquad A \in \As \, .
\end{align}
Consequently  $\pi(\As^\infty) \subset \Cs^\infty$, so one can define 
for any $A,B \in \As^\infty$ the operators $\pi(A)_{\, Q} \in \Cs^\infty$
and the product $\pi(A) \tq \pi(B)$; moreover,
$\pi(A) \tq \pi(B) = \pi(A \tQ B)$.

     After having established the properties of the warping procedure on
$\Cs^\infty$, it is almost evident that the covariant 
representation $(\pi, \Hs)$ of $\As$ induces a covariant representation 
$(\pi^Q , \Hs)$ of $(\As^Q, \tQ)$. It is fixed by setting 
\begin{equation} \label{abstract}
\pi^{\, Q}(A) \doteq \pi(A)_Q \, , \quad A \in \As^\infty \, .
\end{equation}
By Theorem \ref{warpedrep}, the operators $\pi^Q(A)$
are bounded. Moreover, it follows from Lemma \ref{symmetric} that
$$ \pi^Q(A)^* = (\pi(A)_Q)^* = (\pi(A)^*)_Q = \pi(A^*)_Q = \pi^Q(A^*) \, . $$
Similarly, Lemma \ref{rep} implies
$$  \pi^Q(A) \pi^Q(B) = \pi(A)_Q \pi(B)_Q = 
(\pi(A) \tq \pi(B))_Q = \pi(A \tQ B)_Q = \pi^Q(A \tQ B) \, .$$
Finally, one may employ the analogue of Lemma \ref{positive} in the 
abstract setting and the reasoning thereafter to obtain 
$\| \pi^Q (A) \| \leq \| A \|^Q$, $A \in \As^\infty$.  
Hence the homomorphism $\pi^Q : \As^\infty \rightarrow \BH$ can be extended
by continuity to a representation of $(\As^Q, \tQ)$, as claimed. 
>From the first part of Proposition \ref{survey} it follows that
\begin{align*}
& U(x) \pi^Q(A) U(x)^{-1} \\
& = U(x) \pi(A)_Q U(x)^{-1}
= (U(x) \pi(A) U(x)^{-1})_Q = \pi(\alpha_x(A))_Q = \pi^Q(\alpha_x(A))
\, ,
\end{align*}
for all $A \in \As^\infty$. So the representation $\pi^Q$ is also 
covariant, hence $\pi^Q (\As^\infty) \subset \Cs^\infty$.
Depending on the properties of the chosen representation
$(\pi, \Hs)$ of $\As$,  the map $\pi^Q: \As^\infty \rightarrow \Cs^\infty$
may not be injective or surjective. But according to part (ii) of 
the preceding proposition one has $\pi(A)_Q = 0$ if and only if $\pi(A) = 0$, 
$A \in \As^\infty$. Furthermore, in view of the continuity of the 
action $\alpha$ on $\As$ and $\As^Q$, the inclusions 
$\ker\pi \bigcap \As^\infty\subset\ker\pi$ and 
$\ker\pi^Q \bigcap 
\As^\infty\subset\ker\pi^Q$ are dense in the norms $\|\cdot\|$ 
and $\|\cdot\|^Q$, respectively. Thus it follows that $\pi^Q$ is faithful if 
and only if $\pi$ is faithful.
\begin{theorem} \label{faithful} 
Let $(\pi, \Hs)$ be an $\alpha$--covariant representation of the 
$C^*$--algebra $\As$. The homomorphism $\pi^Q : \As^\infty \rightarrow \BH$,
fixed by the relation 
$$
\pi^{\, Q}(A) \doteq \pi(A)_{\, Q} \, , \quad A \in \As^\infty \, ,
$$
extends continuously to an $\alpha$--covariant representation of the 
$C^*$--algebra $(\As^Q, \tQ)$. Moreover, $\pi^Q$ is faithful if and only
if $\pi$ is faithful.
\end{theorem}

     So the warping method provides a representation of the deformed algebras
in the same Hilbert space as the undeformed algebra, enabling the direct
comparison of deformed operators corresponding to different $Q$. This point 
will prove to be useful in the physical context treated below.

\subsection{Further Properties of Warped Convolutions}

     Even though the warped convolutions may be viewed as merely generating 
certain specific representations of Rieffel algebras, it will be 
advantageous to base the subsequent discussion directly on them 
without referring to the Rieffel setting. The reasons for this are threefold: 
(a) It will be necessary to deal with subalgebras of the algebra of 
smooth operators which are not invariant under the automorphic action of the 
translations. So there is no corresponding Rieffel algebra, but the 
warping procedure is still meaningful. (b) It will be necessary
to consider warped operators $A_{Q}, A'_{Q'}$ and their sums and products 
for different matrices $Q, Q'$. Such operations can be carried out 
in the framework of warped convolutions more easily than in the Rieffel setting,
where one has to use Hilbert modules instead of Hilbert spaces.
(c) We shall need to establish algebraic properties of the warped 
operators arising from spectral properties of the unitary representation $U$, 
which  are not available in the Rieffel setting.

     Returning to the Hilbert space framework, we first 
exhibit some general covariance properties of the warped 
convolutions, \cf \cite{BuSu2}. To this end we consider 
(anti)unitary operators $V$ whose adjoint actions on the translations $U$ 
induce linear transformations of $\RR^n$. It follows at once that for any 
such $V$ the algebra $\Cs^\infty$ is stable under the corresponding adjoint 
action, $V \Cs^\infty V^{-1} = \Cs^\infty$, and $V {\cal D} = {\cal D}$. The 
following result is the first instance where we must deal with
warped convolutions for different choices of the underlying matrix $Q$. 
\begin{prop} \label{2.5}
Let  \ $V$ be a unitary or antiunitary  
operator on $\cH$ such that $V U(x) V^{-1} = U(Mx)$,
$x \in \RR^n$, for some invertible matrix $M$. Then, for $A \in \Cs^\infty$,
$$ V  A_Q  V^{-1} = (VAV^{-1})_{\, \sigma MQM^T} \, , $$
where $M^T$ is the transpose of $M$ with respect to the chosen 
bilinear form, $\sigma = 1$ if $V$ is unitary and $\sigma = -1$ if $V$ is 
antiunitary. 
\end{prop}
\begin{proof} Making use of relation (\ref{convenient}) for real $f$, 
one commences from the equalities of strong integrals
\begin{equation*}
\begin{split}
& V \! \iint \! dx dy \, e^{-ixy} \, f(\varepsilon x, \varepsilon y) \, 
\alpha_{Qx}(A) \, U(y) \, V^{-1} \\
& = \iint \! dx dy \, e^{-i\sigma xy} \, f(\varepsilon x, \varepsilon y) \, 
\alpha_ {MQx}(VAV^{-1}) \, U(My) \\
& = \iint \! dx dy \, e^{-ixy} \, f(\varepsilon \sigma M^Tx, \varepsilon M^{-1}y) \, 
\alpha_ {\, \sigma MQM^Tx}(VAV^{-1}) \, U(y) \, ,
\end{split}
\end{equation*}
where the last equality is obtained by substituting  
$(x, y) \rightarrow (\sigma M^Tx, M^{-1}y)$. 
Applying these relations to any vector $\Phi \in {\cal D}$ and taking into 
account $V^{-1} {\cal D} =  {\cal D}$, the assertion follows in the limit 
of small $\varepsilon$.
\end{proof}

     Next, we establish a result which is fundamental for the applications 
to physics. We shall show that the warped convolutions preserve certain 
specific commutation properties of the operators in $\Cs^\infty$ for 
appropriate choices of the underlying skew symmetric matrices depending on the 
spectrum of the representation $U$ \cite{BuSu2}. 

\begin{prop} \label{2.4}
Let $A,B \in \Cs^\infty$ be operators such that
$[\alpha_{Qx}(A), \alpha_{-Qy}(B)] = 0$ for all $x,y \in \mbox{sp} \, U$. Then
$$[ A_Q, B_{-Q}] = 0 \, .$$
\end{prop}

\begin{proof} Returning to the definition of the warped convolutions 
by the spectral calculus and making use of Lemma \ref{symmetric}, one finds for 
vectors $\Phi, \Psi$ with compact spectral support
$$ \langle \Phi, A_Q B_{-Q} \Psi \rangle =
\lim_{F, F' \nearrow 1} \langle \Phi, \big( \! \int \! dE(x) F \alpha_{Qx}(A) \big)
 \big(\! \int \! \alpha_{-Qy}(B) F' dE(y) \big) \Psi \rangle \, ,
$$
where $F, F'$ are finite--dimensional projections. Now
\begin{equation*}
\begin{split} 
 & \langle \Phi, \big( \! \int \! dE(x) F \alpha_{Qx}(A) \big)
 \big(\! \int \! \alpha_{-Qy}(B) F' dE(y) \big) \Psi \rangle \\
& = \iint \langle \Phi,  dE(x) F \alpha_{Qx}(A) \, \alpha_{-Qy}(B) F'
dE(y) \Psi \rangle \\
& = \iint  \langle \Phi,  dE(x) F \alpha_{-Qy}(B) \, \alpha_{Qx}(A) F'
dE(y) \Psi \rangle \, ,
\end{split}
\end{equation*}
where the step from the first to the second line is justified
by the fact that the given expression can be decomposed into a finite
sum of product measures multiplied with smooth functions.
The second step is a consequence of the commutativity properties of
$A$ and $B$. Introducing the notation 
$\boldsymbol{u} = (v,w,x,y) \in \RR^{4n}$ and picking any test
function $\boldsymbol{u} \mapsto h(\boldsymbol{u})$ which 
which is equal to $1$ at $0$, it follows from the spectral
representation of $U$ that the latter 
integral is equal to
$$ 
\lim_{\varepsilon \rightarrow 0} \ (2 \pi)^{-2n} \! 
\iiiint \! d\boldsymbol{u} \, h(\varepsilon \boldsymbol{u})\, 
e^{-ivx -iyw} \,
 \langle \Phi,  U(v) F \alpha_{-Qy}(B) \, \alpha_{Qx}(A) F' U(w) \Psi 
 \rangle \, . $$
Adopting now the arguments and notation in the final part of the proof of
Lemma \ref{rep}, one finds that for the polynomial 
$L$ \eqref{polynomial-L} there exists a 
corresponding polynomial $P$ such that
\begin{equation*}
\begin{split} 
& \iiiint \! d\boldsymbol{u} \, h(\varepsilon \boldsymbol{u}) \, e^{-ivx -iyw} \,
 \langle \Phi,  U(v) F \alpha_{-Qy}(B) \, \alpha_{Qx}(A) F' U(w) \Psi \rangle \\
= & \iiiint \! d\boldsymbol{u} \,  e^{-ivx -iyw} \,
P(\boldsymbol{\partial}) \, h(\varepsilon \boldsymbol{u}) \,
L(\boldsymbol{u})^{-1} 
 \langle \Phi,  U(v) F \alpha_{-Qy}(B) \, \alpha_{Qx}(A) F' U(w) \Psi \rangle \, .
\end{split}
\end{equation*}
After having performed the differentiations in the last integral, one sees 
by an application of the dominated convergence theorem that the composite limit
$\varepsilon \rightarrow 0$, $F, F' \nearrow 1$ is independent of the order in 
which the individual limits are carried out and also does not depend on the 
choice of $h$. Thus one has, in particular,
\begin{equation*}
\begin{split}  
& \langle \Phi, A_Q B_{-Q} \Psi \rangle \\
& = \lim_{\varepsilon \rightarrow 0} \ (2 \pi)^{-2n}
\iiiint \! d\boldsymbol{u} \,  h(\varepsilon \boldsymbol{u}) \,
e^{-ivx -iyw} \,
 \langle \Phi,  U(v) \alpha_{-Qy}(B) \, \alpha_{Qx}(A) U(w) \Psi \rangle \, .
\end{split}
\end{equation*}
As before, one takes advantage of the fact that the integration
may be restricted to the submanifold 
$(\mbox{ker} \, Q)^\perp \times \dots \times (\mbox{ker} \, Q)^\perp
\subset \RR^{4n}$, since the remaining integrals merely produce
factors of $2 \pi$. So the preceding integral can be recast as
\begin{equation*}
\begin{split}
& \iiiint \! d\boldsymbol{u} \,  e^{-ivx -iyw} \,
h(\varepsilon\boldsymbol{u}) \,
 \langle \Phi,  U(v) \alpha_{-Qy}(B) \, \alpha_{Qx}(A) U(w) \Psi \rangle \\
& = \iiiint \! d\boldsymbol{u} \,  e^{-ivx -iyw} \,
h(\varepsilon\boldsymbol{u}) \,
 \langle \Phi,  U(w) \alpha_{-Qy+v-w}(B) \, \alpha_{Qx+v-w}(A) U(v) \Psi \rangle \\
& = \iiiint \! d\boldsymbol{u} \,  e^{-ivx -iyw} \,
k(\varepsilon\boldsymbol{u}) \,
 \langle \Phi,  U(w) \alpha_{-Qy}(B) \, \alpha_{Qx}(A) U(v) \Psi \rangle \, ,
\end{split}
\end{equation*}
where 
$k(v,w,x,y) = h(v,w,x - Q^{-1}(v - w),y + Q^{-1}(v - w))$.
The last equality is the result of the substitution  
$(x,y) \rightarrow (x - Q^{-1}(v - w), y + Q^{-1}(v - w)$,
under which $ e^{-ivx -iyw} $ does not change because of the
skew symmetry of $Q$. Proceeding to the limit of small $\varepsilon$,
one obtains by relation (\ref{convenient}) and Lemma \ref{symmetric} 
\begin{equation*}
\begin{split}  
& \lim_{\varepsilon \rightarrow 0} \ (2 \pi)^{-2n}
\iiiint \! d\boldsymbol{u} \,  e^{-ivx -iyw} \,
k(\varepsilon\boldsymbol{u}) \,
 \langle \Phi,  U(w) \alpha_{-Qy}(B) \, \alpha_{Qx}(A) U(v) \Psi \rangle \\
= & \ \langle \Phi,  B_{-Q} A_Q \Psi \rangle \, .
\end{split}
\end{equation*}
This shows that 
$\langle \Phi, A_Q B_{-Q} \Psi \rangle = \langle \Phi,  B_{-Q} A_Q \Psi \rangle$. 
Since $\Phi, \Psi$ were arbitrary elements of a dense set of vectors, 
the assertion now follows.
\end{proof}

     We finally discuss the structure of the family of maps given by the 
warped convolutions. According to Proposition \ref{survey} (ii), these 
maps act bijectively on $\Cs^\infty$ and therefore can be composed 
and have inverses.
In fact, they form a group which is homomorphic to 
$\RR^{n(n-1)/2}$, as can be seen from the next proposition. 

\begin{prop}  \label{group}
Let $Q_1$, $Q_2$ be skew symmetric matrices. Then 
$$(A_{Q_1})_{Q_2} = A_{Q_1 + Q_2} \, , \quad A \in \Cs^\infty \, .$$
\end{prop}

\begin{proof}
To begin, note that for any continuous bounded function $f$
of the generator $P$ of $U$ one has $A_Q f(P) = (A f(P))_Q$, as a consequence 
of relation (\ref{covariance}) and part (iv) of Lemma \ref{1.1}. Let 
$\Phi \in \Ds$ be any vector with compact spectral support and let $f$ be a 
test function  such that $f(P) \Phi = \Phi$. It follows that  
$(A_{Q_1})_{Q_2} \Phi = (A_{Q_1} f(P) )_{Q_2} f(P) \Phi$. 
Picking nets of finite--dimensional projections $F, F^\prime$
converging to $1$, making use of the spectral calculus, which implies 
$f(P) \, dE(z) = f(z) \, dE(z)$, $z \in \RR^n$, and 
recalling the definition of the warped convolutions, one obtains
in the sense of weak convergence 
$$
(A_{Q_1})_{Q_2} \Phi =
\lim_{F^\prime \nearrow 1, F \nearrow 1} 
\iint \!  f(x) f(y) \, \alpha_{Q_1x  + Q_2y}(A) \,  F dE(x) F^\prime dE(y)
\Phi \, .
$$
Here the limits are taken in the given order and the  
(strong) limit $F \nearrow 1$ has been interchanged with the 
$y$--integration by an application of the dominated convergence theorem.
Since the function 
$x,y \mapsto  f(x) f(y) \alpha_{Q_1x  + Q_2y}(A)$ is smooth
and rapidly decreasing in norm, one can interchange the limits.
The product measure $dE(x) F^\prime dE(y)$ converges 
weakly in the sense of distributions to $\delta(x-y) \, dx dE(y)$
as $F^\prime \nearrow 1$, where $\delta(x-y) \, dx$ is the Dirac measure at 
$y$; hence one obtains 
$$
(A_{Q_1})_{Q_2} \Phi =
\lim_{F \nearrow 1} 
\int \!  \alpha_{(Q_1 +  Q_2)x}(A) \,  F  dE(x) f(P)^2 \Phi =  A_{Q_1 + Q_2} \Phi \, .
$$
The desired conclusion then follows, because the space of vectors $\Phi$ with 
compact spectral support is dense in $\Hs$.  
\end{proof} 

\noindent Note that this result does not entail a composition law
of the representations $\pi_Q$ of the Rieffel algebras, since
their ranges do not, in general, fit with their respective domains. \\[5mm]
\textbf{Further Results:} Most of the preceding results can be
established in a setting of unbounded operators. One proceeds again from 
a continuous unitary representation $U$ of $\RR^n$ and considers the 
${}^*$--algebra ${\cal F}$ of \textit{all}  operators $F$ for which there is 
some $n_F \in \NN$ such that the functions 
$x \mapsto (1+P^2)^{-n_F} \alpha_x(F) (1+P^2)^{-n_F}$
are arbitrarily often differentiable in norm. The operators $F \in {\cal F}$ 
are defined on the domain ${\cal D}$ and leave it invariant. Making use of 
the fact that there is a version of Lemma \ref{1.1} in this setting, one can 
define the Rieffel product $\tq$ on  ${\cal F}$; the warped convolutions of the 
elements of ${\cal F}$ can be defined as well and are elements of ${\cal F}$. 
Moreover, Lemmas \ref{symmetric} and \ref{rep} hold without changes, so
the warped convolutions define an (unbounded) ${}^*$--representation of
$(\Fs,\tq)$, and Propositions \ref{2.5}, \ref{2.4} 
and~\ref{group} hold as well. We refrain from giving the proofs here.

\section{Warped Convolutions and Borchers Triples}  \label{triple}
\setcounter{equation}{0}

     We consider now warped convolutions in the context of Borchers triples, 
invented by Borchers \cite{Bo} for the construction and analysis of 
relativistic quantum field theories. This setting is, on the one hand, 
more restrictive than the preceding one, since one deals with unitary 
representations $U$ of the translations $\RR^n$, $n \geq 2$, with certain 
specific spectral properties. On the other hand, one considers 
subalgebras of $\BH$ on which the adjoint action $\alpha$ of $U$ merely induces 
endomorphisms for semigroups of translations in the set  
$\Ws \doteq \{ x = (x_0, x_1, \ldots, x_{n-1}) \in \RR^n : x_1 \geq \vert
x_0 \vert \} $.
\begin{definition}  \label{deftriple}
A Borchers triple $({\cal R}, U, \Omega)$ (relative to $\Ws$) consists of
\begin{itemize}
\item[(a)] a von Neumann algebra ${\cal R} \subset \BH$, 
\item[(b)] a weakly continuous unitary
  representation $U$ of 
${\mathbb R}^n$ on ${\cal H}$ whose spectrum is contained in the 
closed forward light cone 
$V_+ = \{p = (p_0, p_1, \ldots, p_{n-1}) \in~\RR^n : 
p_0 \geq  \sqrt{p_1^2 + \ldots  + p_{n-1}^2}  \}$ 
and which satisfies 
$\alpha_x(\Rs)  \subset \Rs$, $x \in \Ws$,
\item[(c)] and a unit vector $\Omega \in {\cal H}$ which is invariant under 
the action of $U$ and is cyclic and separating for ${\cal R}$.
\end{itemize}
\end{definition}
By condition (c), Tomita--Takesaki theory \cite{Ta, Ta2} 
is applicable to the pair $({\cal R}, \Omega)$, and we 
shall denote by $\Delta, J$ the associated modular operator and
involution. In this context 
Borchers \cite{Bo} proved the following remarkable theorem (see \cite{Fl}
for a simpler proof). 

\begin{theorem} \label{Borchers}
Let $({\cal R}, U, \Omega)$ be a Borchers triple relative to $\Ws$.
Denoting by $\vartheta(t)$, $t \in \RR$, 
and $j$ the
transformations acting  on 
$x =  (x_0, x_1, \dots x_{n-1}) \in \RR^n$ by
\begin{align*} 
\vartheta(t) \, x &  \doteq (\cosh(2\pi t) x_0 + \sinh(2\pi t) x_1, \sinh(2\pi t) x_0 + \cosh(2\pi t) x_1,
x_2, \dots , x_{n-1}) \, , \\
j x  & \doteq (-x_0, -x_1, x_2, \dots x_{n-1}) \, ,
\end{align*} 
one has
\begin{enumerate} 
\item[(i)] $\Delta^{it} U(x) \Delta^{-it} = U(\vartheta(t) x)$
for $x \in \RR^n$ and  $t \in \RR$, 
\item[(ii)] $J U(x) J = U(j x)$
for $x \in \RR^n$.
\end{enumerate}
Moreover, $(\Rs^\prime, U, \Omega)$ is a Borchers triple relative 
to $- \Ws$, where $\Rs^\prime = J \Rs J$ is the commutant of $\Rs$.
\end{theorem}

\begin{proof}
The assertion for $n = 2$ is proven in \cite{Bo}. 
Setting $x_\perp = (0,0,x_2,\ldots,x_{n-1})$, conditions (b), (c) 
in Definition \ref{deftriple} imply 
$U(x_\perp) \Rs U(x_\perp)^{-1} = \Rs$ and $U(x_\perp) \Omega =
\Omega$.
The uniqueness of the modular objects then entails that $\Delta$ and $J$ 
both commute with all $U(x_\perp)$, completing the proof in 
the general case.
\end{proof}

     We shall show now that the family of Borchers triples is stable
under the deformations induced by warped convolutions corresponding to 
certain  specific choices of the skew symmetric 
matrix\footnote{\label{Lorentz} Having in mind
applications to quantum field theory, we choose henceforth the Lorentz 
product $xy = x_0 y_0 - \sum_{m=1}^{n-1} x_m y_m $, $x,y \in \RR^n$, as 
the bilinear form on $\RR^n$.} $Q$. Moreover, the
modular objects of the deformed triples coincide with those of the
original one. This observation is of relevance in quantum
field theory, which will be discussed at the end of this section.

     We begin with some technical remarks. Let $\Cs^\infty$ be, as above, 
the ${}^*$--algebra of all smooth elements in $\BH$ under the adjoint action 
of the translations and let $\Rs^\infty = \Rs \bigcap \Cs^\infty$.  In view of 
condition (b) in Definition \ref{deftriple}, one obtains elements of 
$\Rs^\infty$ by smoothing any element $R \in {\cal R}$ with Schwartz test 
functions $f$ having support in ${\cal W}$,  
\begin{equation} \label{smoothing}
R(f) \doteq \int \! dx \, f(x) \, \alpha_x(R) \, .
\end{equation}
These weak integrals are elements of $\Rs^\infty$ since, by construction, 
they are smooth and contained in the von Neumann algebra $\Rs$. Choosing 
sequences $f_n$  of test functions with support in ${\cal W}$ which 
approximate the Dirac measure at $0$, one sees that $\Rs^\infty$ is dense in 
$\Rs$ in the strong operator topology, and consequently $\Omega$ is cyclic 
for $\Rs^\infty$. By the same reasoning one finds that $\Omega$ is also
cyclic for $\Rs^{\prime \, \infty} \doteq \Rs^\prime \bigcap \Cs^\infty$.

     Now let $Q$ be any real skew symmetric matrix on $\RR^n$  
which is ${\cal W}$--{\em compatible} in the sense that
$Q V_+ \subset {\cal W}$. This constraint on $Q$ will become important in 
the following. The corresponding warped operators $A_Q$, $A \in \Rs^\infty$, are
defined as in the previous section. Since they are bounded and satisfy 
${A_Q}^* = {A^*}_Q$, they generate a von Neumann algebra, called a warped 
algebra for short. With a slight abuse of notation, we write  
$$ {\cal R}_Q \doteq \{A_Q : A \in \Rs^\infty \}^{\prime \prime} \, . $$

     For the proof that the warped triple $({\cal R}_Q, U, \Omega)$ 
is again a Borchers triple, we note that, as a consequence of 
Proposition \ref{2.5}, one has 
$\alpha_x( {\Rs}_Q) = \alpha_x({\Rs})_Q \subset {\Rs}_Q$
for~$x \in {\Ws}$. So condition (b) in Definition~\ref{deftriple} is 
satisfied. Furthermore, since $\Omega$ is cyclic for $\Rs^\infty$, it 
is also cyclic for ${\cal R}_Q$ as a consequence of equation 
\eqref{vacuum}. In order to see that $\Omega$ is separating for ${\cal R}_Q$,
let $A \in \Rs^\infty$, $A^\prime \in \Rs^{\prime \, \infty}$. Then   
$[\alpha_x(A) , \alpha_y(A^\prime)] = 0$ for $x \in {\cal W}$, 
$y \in - {\cal W}$, and taking into account that 
$Q \, \mbox{sp} \, U \subset Q \, V_+ \subset \Ws$, it follows from
Proposition \ref{2.4} that $[A_Q, {A^\prime}_{-Q}] = 0$. Thus 
$(\Rs^\prime)_{-Q} \subset {({\cal R }_Q)}^\prime$. 
But equation \eqref{vacuum} implies that $\Omega$ is cyclic for
$(\Rs^\prime)_{-Q}$ and thus \textit{a fortiori} for~${({\cal R }_Q)}^\prime$.
Hence $\Omega$ is separating for ${\cal R}_Q$,
and condition (c) in Definition~\ref{deftriple} holds as well.

\begin{theorem}  \label{qtriple}
Let $(\Rs,U,\Omega)$ be a Borchers triple relative to $\Ws$ and
let $Q$ be $\cW$--compatible. Then the resulting 
warped triple $(\Rs_Q,U,\Omega)$ is also a Borchers triple relative to $\Ws$.
\end{theorem}

    In view of this theorem, we may apply modular theory to the warped 
triple. We shall show next that the corresponding modular objects coincide 
with the original ones. To this end we need the following technical lemma.

\begin{lemma} \label{3.4}
Let $(\Rs, U, \Omega)$ be a Borchers triple relative to ${\cal W}$
and let \mbox{$S = J \Delta^{1/2}$} be the corresponding Tomita conjugation 
given by the closure of the map
$$  S A  \, \Omega = A^*  \, \Omega \, , \quad A \in \Rs \, . $$
Then the subdomain $\Rs^\infty \, \Omega$  is a core for $S$.
\end{lemma}
\begin{proof}
Let $R \in {\cal R}$ and let  $f_n$ be a sequence of real test
functions with support in ${\cal W}$ such that in the sense of strong
convergence $\lim_n  R(f_n) \, \Omega = R \, \Omega$, \cf
relation~(\ref{smoothing}) and the remarks thereafter. Since  
$ R(f_n) \in \Rs^\infty \subset \Rs$ and 
$$ \lim_n S  R(f_n)  \, \Omega = \lim_n R(f_n)^* \, \Omega =
\lim_n R^*(f_n)  \, \Omega =  R^*  \, \Omega
=  S R \Omega \, , $$
the conclusion follows, because $\Rs \Omega$ is a core for $S$ 
by definition.
\end{proof}

     We are now in a position to establish the invariance of the 
modular objects of Borchers triples under the warping procedure. 

\begin{theorem} \label{modular}
Let $({\cal R}, U, \Omega)$ be a Borchers triple relative to $\Ws$ with 
modular objects $\Delta, J$, and let $Q$ be a $\Ws$--compatible 
matrix. Then the modular objects $\Delta_Q, J_Q$ associated with the warped 
triple $({\cal R}_Q, U, \Omega)$ coincide with those of the original 
triple, \ie
$$ \Delta_Q = \Delta, \quad J_Q = J \, . $$
\end{theorem}
\begin{proof} Let $S_Q$ be the Tomita conjugation associated with 
the warped triple $({\cal R}_Q, U, \Omega)$ and let $S$ be the
Tomita conjugation associated with $({\cal R}, U, \Omega)$. Since 
$A_Q \in \Rs_Q$ for $A \in \Rs^\infty$, equation \eqref{vacuum}
and Lemma \ref{symmetric} imply 
$$ 
S_Q \, A \Omega = S_Q \, A_Q \Omega = (A_Q)^* \, \Omega = (A^*)_Q \, \Omega 
= A^* \, \Omega = S A \Omega \, .
$$
According to the preceding lemma, $\Rs^\infty \Omega$ is a core for
$S$, hence $S_Q \supset S$. By the Tomita--Takesaki theory
\cite{Ta,Ta2}, the adjoint ${S_Q}^*$ of $S_Q$ is the Tomita conjugation
associated with $((\Rs_Q)^\prime, U, \Omega)$, and similarly
$S^*$ is the Tomita conjugation associated with $(\Rs^\prime, U, \Omega)$.
It was shown in the proof of Theorem \ref{qtriple} that
$(\Rs^\prime)_{-Q} \subset {({\cal R }_Q)}^\prime$. Thus, as
${A^\prime}_{-Q} \in {\Rs^\prime}_{-Q}$ for $A^\prime \in \Rs^{\prime
  \, \infty}$, one obtains by another application of equation 
\eqref{vacuum} and Lemma \ref{symmetric}
$$
{S_Q}^* \, A^\prime \, \Omega = {S_Q}^* \, {A^\prime}_{-Q} \, \Omega
=  ({A^\prime}_{-Q})^* \, \Omega =  (A^{\prime \, *})_{-Q} \, \Omega
= A^{\prime \, *} \, \Omega = S^* \, A^\prime \, \Omega \, .
$$
By the preceding lemma $\Rs^{\prime \, \infty} \, \Omega$ is a core
for $S^*$, hence ${S_{Q}}^* \supset S^*$ and consequently  
$S \supset S_Q$, since both conjugations are closed operators.
Thus $S_Q = S$ and, by the uniqueness of the polar decomposition, the 
desired conclusion follows. 
\end{proof}

     An immediate consequence of this theorem is the observation that 
\begin{align}\label{RQ-commutant}
 {\Rs_Q}' = {\Rs'}_{-Q}\,.
\end{align}
     Indeed, Theorem \ref{Borchers} and Proposition~\ref{2.5}  
imply  $J \Rs_Q J = (J \Rs J)_{-jQj}$ and it is also straightforward to 
verify that $j Q j = Q$ for any $\Ws$--admissible matrix $Q$. 
Since $J = J_Q$, the asserted equation then follows from Tomita--Takesaki
theory.

     Let us discuss now the physical significance of these findings.
As was pointed out in \cite{Bo}, Theorem \ref{Borchers} allows one to use 
the Borchers triple $(\Rs, U, \Omega)$ as a building block for the 
construction of a quantum field theory in two spacetime dimensions. 
Identifying the cone ${\Ws} \subset \RR^2$ defined above with the corresponding 
wedge shaped region in two--dimensional Minkowski space, one interprets 
$\As(\Ws) \doteq \Rs$ as the algebra generated by observables
which are localized in $\Ws$. Moreover, noticing that the transformations 
$\vartheta(t)$, $t \in \RR$, and $j$ introduced in Theorem \ref{Borchers}
have the geometrical meaning of Lorentz boosts and spacetime
reflection, respectively, one can consistently extend the representation
$U$ of the translations $\RR^2$ to a continuous (anti)unitary representation 
of the proper Poincar\'e group $\Ps_+$. It is given by 
$$ 
U(\lambda) \doteq U(x)  J^\sigma  \Delta^{it}  \, ,
\quad \lambda = (x, j^\sigma  \vartheta(t)) \in \Ps_+ \, ,
$$  
where $x \in \RR^2$, $t \in \RR$ and $\sigma \in \{0,1 \}$. 
Thus $J$ represents the PCT--operator. With the help of this representation 
one can define the algebras generated by observables in the 
transformed wedge regions~$\lambda \Ws $, $\lambda \in \Ps_+$ by setting
$$ 
\As(\lambda  \Ws) \doteq U(\lambda) \Rs U(\lambda)^{-1} \, , \quad
\lambda \in \Ps_+ \, .
$$
This definition is consistent, since the stability group of the 
wedge $\Ws$ in $\Ps_+$ consists of the boosts $\vartheta(t)$,
$t \in \RR$, whose corresponding automorphic action leaves the 
algebra $\Rs$ invariant according to Tomita--Takesaki theory. The resulting 
assignment $\Ws_{.} \mapsto \As(\Ws_{.})$ of wedge regions to algebras defines 
a net (pre--cosheaf) on $\RR^2$. It is Poincar\'e covariant by
construction and causal. In fact, since $j$ maps the 
wedge $\Ws$ onto its spacelike complement $\Ws^\prime=-\Ws$, one has
$$
\As(\Ws^\prime) = U(j) \As(\Ws) U(j)^{-1} = J \Rs J = \Rs^\prime =
\As(\Ws)^\prime \, ,
$$
where the third equality follows from Tomita--Takesaki theory.
So the observables in spacelike separated wedges commute, in 
accordance with the principle of Einstein causality.
In this way any Borchers triple defines a relativistic quantum 
field theory in two spacetime dimensions, \cf \cite{Bo} for more 
details.

     The upshot of these considerations is the insight that, as a consequence 
of the preceding three theorems, the warped triples $(\Rs_Q, U, \Omega)$ 
generate in the same manner another causal and covariant net  
\mbox{$\Ws_{.} \mapsto \As_Q(\Ws_{.})$} by setting 
$$ 
\As_Q(\lambda  \Ws) \doteq U(\lambda) \Rs_Q U(\lambda)^{-1} \, , \quad
\lambda \in \Ps_+ \, .
$$
Thus the warping procedure provides a tool for the consistent
deformation of two--dimensional quantum field theories 
without changing the underlying representation of the Poincar\'e
group. We shall further elaborate on this observation in the next  
section.

\section{\hspace*{-5mm} 
Warped Convolutions in  Quantum Field Theory}  \label{application}     
\setcounter{equation}{0}

     In this section we examine applications of the warping procedure to
relativistic quantum field theories in more than two spacetime
dimensions. Thus we interpret $\RR^n$,~$n > 2$, as 
Minkowski space equipped with the standard metric fixed by the 
Lorentz product, \cf footnote \ref{Lorentz}. The
identity component of its isometry group, the 
Poincar\'e group, is the semidirect product  $\Pid = \RR^n \rtimes \Lid$ 
of the spacetime translations $\RR^n$ and the proper orthochronous Lorentz
transformations $\Lid$.

     In a manner similar to  the preceding section, we 
describe the theories in the algebraic setting of local quantum 
physics \cite{Haag} by a qualified version of the concept of 
Borchers triple. Additional constraints arise since, 
on the one hand, the group generated by the translations, along with the 
boosts and reflection emerging from the modular structure 
of the triple,  does not act transitively on the set of wedge regions 
in $\RR^n$ if $n > 2$. The smallest subgroup of the Poincar\'e group which 
fulfills this condition is  $\Pid$. So one needs from the outset an 
action of this group on the underlying algebra $\Rs$, which one interprets 
again as the algebra of observables localized  in the given wedge 
region 
$\Ws \doteq \{ x = (x_0, x_1, \ldots, x_{n-1}) \in \RR^n : x_1 \geq \vert
x_0 \vert \} $. 
On the other hand, one must  ensure that this action is
consistent with the principle of Einstein causality, 
according to which observables in spacelike separated regions
must  commute. The resulting consistency conditions can be expressed
in terms of the triple in an evident manner,  \cf
\cite[Proposition 7.3.22]{BaWo}. They lead to the concept of 
a causal Borchers triple.
\begin{definition} \label{cBorchers}
A causal Borchers triple $(\Rs, U, \Omega)$ relative to $\Ws$ consists of 
\begin{itemize}
\item[(a)] a von Neumann algebra $\Rs \subset \BH$, 
\item[(b)] a weakly continuous unitary representation $U$ of $\Pid$
such that, $\lambda \in \Pid$,   
\begin{align*}
U(\lambda) \Rs U(\lambda)^{-1} & \subset \Rs \ \ \mbox{if} \ \ \lambda \Ws \subset
\Ws  \, , \\
U(\lambda) \Rs U(\lambda)^{-1} & \subset \Rs^\prime 
\ \ \mbox{if} \ \ \lambda \Ws \subset \Ws^\prime,
\end{align*}
and the spectrum of the abelian subgroup $U \rest \RR^n$ of the spacetime 
translations is contained in the closed forward lightcone $V_+$,
\item[(c)] and a unit vector $\Omega \in \Hs$, 
describing the vacuum, which is invariant under the
  action of $U$ and is cyclic and separating for $\Rs$.
\end{itemize}
\end{definition}

\noindent \textbf{Remark} \ In two spacetime dimensions any Borchers triple
determines a causal Borchers triple by the modular construction in the 
preceding section. As there is no element in $\Pid$ which maps the 
wedge $\Ws$ into its spacelike (causal) complement $\Ws^\prime$,
the second constraint in condition (b) is trivially satisfied in this case.
In order to flip the wedge one needs the spacetime reflection $j$,
which is an element of $\Ps_+ \supset \Pid$. As we have seen, its corresponding 
action on $\Rs$ is consistent with Einstein causality as a consequence of 
modular theory. In higher dimensions one either has to posit 
causality from the outset, as we do, or one has to impose additional 
constraints on the modular structure of the triple which imply   
it, \cf \cite{BrGuLo2,BDFS,BuSuIII,Gu,Bo1}.

\vspace*{1mm}
     With the above input one can define the algebras corresponding to 
arbitrary regions in $\RR^n$ in a straightforward manner, which we 
briefly recall. Making use of the fact that $\Pid$ acts transitively on the 
wedge regions, one begins with the wedge algebras by setting 
\begin{equation} \label{wedgealgebras}   
\cA(\lambda \cW) \doteq U(\lambda) \Rs U(\lambda)^{-1} \, ,
\quad \lambda \in \Pid \, .
\end{equation}
This definition is consistent, since $\lambda_1 \cW = \lambda_2 \cW$ implies 
that the transformation $\lambda_2^{-1}  \lambda_1$ is an element of the 
stability group of $\cW$, and $\cR$ is stable under the adjoint action of the 
corresponding unitary operators according to the first part of
condition (b). Similarly, if $\lambda_1 \cW \subset \lambda_2 \cW$, 
it follows  that 
$U(\lambda_2^{-1}  \lambda_1) \cR U(\lambda_2^{-1}  \lambda_1)^{-1} \subset \cR$,
hence  $\cA(\lambda_1 \cW) \subset \cA(\lambda_2 \cW)$. Thus the 
family of wedge algebras complies with the condition of isotony. 
The wedge algebras also transform covariantly under the adjoint
action of the representation $U$ by their very definition.
Moreover, if $\lambda_1 \cW \subset (\lambda_2 \cW)^\prime$, then  
$U(\lambda_2^{-1}  \lambda_1) \cR U(\lambda_2^{-1}  \lambda_1)^{-1} 
\subset \cR^\prime$
according to the second part of condition (b). Hence
$\cA(\lambda_1 \cW) \subset \cA(\lambda_2 \cW)^\prime$
in accordance with Einstein causality. The algebras corresponding to 
arbitrary causally closed convex regions $\cO \subset \RR^n$ are determined 
from the wedge algebras $\cA(\cW_{\cdot})$ by setting  
$\cA(\cO) = \bigcap_{\cW_{\cdot} \supset \cO} \, \cA(\cW_{\cdot})$. 
It is apparent that the resulting assignment $\cO \mapsto \cA(\cO)$ 
inherits the structure of a causal and covariant net on $\RR^n$, 
\ie of a local quantum theory \cite{Haag}. It should be noted, however, that 
within the present general framework the algebras
corresponding to bounded regions may happen to be trivial. 
We shall comment on the physical significance of this 
possibility at the end of this section.

     We now want to use our warping procedure to deform causal Borchers 
triples. Additional constraints on the underlying skew symmetric matrices arise
due to the extra conditions imposed on such triples. In fact, $Q$ must 
have the following form with respect to the coordinates chosen in the
specification of the wedge $\cW \subset \RR^n$:  
\begin{equation} \label{q1}
Q \doteq 
\left( \begin{array}{ccccc}  0 & \zeta & 0 & \cdots & {0} \\ 
                       \zeta &    0   & 0 & \cdots & {0} \\
                            0 &    0   & 0 & \cdots & {0} \\    
                       \vdots & \vdots & \vdots & \ddots &  \vdots \\
                       {0} &    {0}   & {0} & \cdots & {0} \end{array} \right)
\end{equation}
for fixed $\zeta \geq 0$. In the special but physically most 
interesting case of $n = 4$ dimensions, one can admit matrices of the more 
general form
\begin{equation*}  \label{q2}
Q \doteq 
\left( \begin{array}{cccc}  0 & \zeta & 0 & 0 \\ 
                          \zeta   &  0 & 0 & 0 \\
                            0 &    0   & 0 & \eta \\
                            0 &    0   & -\eta & 0 \end{array} \right)
\end{equation*}
for fixed $\zeta \geq 0$, $\eta \in \RR$. Note that these matrices are 
skew symmetric with respect to the Lorentz product. 
The following facts pointed out in \cite{GrLe} are crucial for the 
consistent deformation of the triples and, in turn, determine the
choice of the admissible matrices~$Q$ \ \cite[Lemma 2]{GrLe}. 

\pagebreak
\begin{enumerate}
\item[(i)] \ $Q \, V_+ \subset \cW $. 
\item[(ii)] 
\ Let $\lambda = (x,\Lambda) \in \Pid$ be such that $\lambda \cW \subset
\cW $. Then $\Lambda Q \Lambda^T = Q$. 
\item[(iii)] 
\ Let $\lambda = (x, \Lambda) \! \in \! \Pid$ 
be such that $\lambda \, \cW \! \subset \! \cW^\prime $. Then 
$\Lambda Q \Lambda^{T} \! = \! - Q$.
\end{enumerate}
Any matrix $Q$ with these properties is said to be 
$\cW$--{\em admissible} (qualifying the notion of $\cW$--compatibility
introduced in the preceding section). 

     Given a causal Borchers triple $(\cR, U, \Omega)$ relative to $\cW$,
we proceed as in the preceding section and define for fixed $\cW$--admissible 
matrix $Q$ the warped von Neumann algebra 
$$ \cR_Q \doteq \{ A_Q : 
A \in \cR^\infty \}^{\prime \prime} \, . $$
The corresponding warped triple $(\cR_Q, U, \Omega)$
is again a causal Borchers triple. For the proof of this fact we make use of 
Proposition \ref{2.5}, according to which
$$ 
U(\lambda) A_Q U(\lambda)^{-1} = (U(\lambda) AU(\lambda)^{-1})_{\Lambda Q \Lambda^T} \, ,
\quad \lambda = (x, \Lambda) \in \Pid \, ,
$$
for all  $A \in \cC^\infty$. Taking into account properties (ii) and 
(iii) of $Q$ given above, we conclude that 
\begin{align*}
U(\lambda) \, \cR_Q \, U(\lambda)^{-1} = (U(\lambda) \, \cR \, U(\lambda)^{-1})_Q
\subset \cR_Q \quad & \mbox{if}  \quad \lambda \cW \subset \cW \, , \\
U(\lambda) \, \cR_Q \, U(\lambda)^{-1} 
= (U(\lambda) \, \cR \, U(\lambda)^{-1})_{-Q}
\subset (\cR^\prime)_{-Q} \quad & \mbox{if}  \quad \lambda \cW \subset \cW^\prime \, .
\end{align*}
But from equation \eqref{RQ-commutant} one has  
$(\cR^\prime)_{-Q} = (\cR_Q)^\prime$; \ hence the warped triple satisfies condition 
(b) in Definition \ref{cBorchers}. In the proof of Theorem \ref{qtriple}, it 
was shown that $\Omega$ is cyclic and separating for $\cR_Q$,  
so the triple also complies with condition~(c).
\begin{theorem} \label{net} 
Let $(\cR, U, \Omega)$ be a causal Borchers triple relative to $\cW$
and let $Q$ be a $\cW$--admissible matrix. The corresponding warped
triple $(\cR_Q, U, \Omega)$ is again a causal Borchers triple 
relative to  $\cW$.
\end{theorem}
 
    So the deformations induced by the warped convolutions 
are consistent with the basic principles of local quantum physics. 
It is noteworthy that also certain more specific features persist under 
these deformations, such as the physically significant 
property of wedge duality. This property can be encoded 
into a maximality condition on the Borchers triple, which implies that
the underlying algebra cannot be enlarged without coming into
conflict with causality. 
\begin{definition}
Let $(\cR, U, \Omega)$ be a causal Borchers triple relative to $\cW$.
The triple is said to be maximally causal   
if \ $U(\lambda) \cR U(\lambda)^{-1} = \cR^\prime$ 
for any $\lambda \in \Pid$ such that 
$\lambda \cW = \cW^{\, \prime}$.  
\end{definition} 
It immediately follows from the definition of the wedge algebras that,    
under these circumstances,   
$\cA(\cW_{\cdot}{}^\prime) = \cA(\cW_{\cdot})^\prime$ 
for all wedges $\cW_{\cdot}$, \ie wedge duality obtains.
\begin{prop}
Let $(\cR, U, \Omega)$ be a maximally causal Borchers triple relative
to $\cW$ and let $Q$ be a $\cW$--admissible matrix.
Then the corresponding warped Borchers triple $(\cR_Q, U, \Omega)$ 
is also maximally causal. 
\end{prop}
\begin{proof}
Let $\lambda\in\Pid$ with $\lambda\Ws=\Ws'$. Then property (iii) of the 
$\Ws$--admissible matrix $Q$, Proposition~\ref{2.5} and the maximality 
condition imply that
$$ U(\lambda) \, \cR_Q \, U(\lambda)^{-1}
= (U(\lambda) \, \cR \, U(\lambda)^{-1})_{-Q}
= \cR^\prime{}_{-Q} \, .$$ 
Equation \eqref{RQ-commutant} completes the proof. 
\end{proof}

     Let us turn now to the question whether the deformed Borchers triples 
generate new theories. It is apparent that equivalent triples, as defined 
below, give rise to isomorphic nets of observable algebras and therefore 
must be identified.
\begin{definition} Let $(\Rs_1, U_1, \Omega_1)$ and $(\Rs_2, U_2, \Omega_2)$
be two causal Borchers triples. The triples are equivalent if there exists
an isometry $V : \Hs_1 \rightarrow \Hs_2$ between the underlying Hilbert 
spaces such that $V \Rs_1 = \Rs_2 V$, $V U_1(\lambda) = U_2(\lambda) V$ 
for all $\lambda \in \Pid$, and $V \Omega_1 = \Omega_2$.
\end{definition}

\noindent Note that the algebras encountered in Borchers triples 
are generically isomorphic to the unique hyperfinite factor of type III${}_1$ 
and hence to each other. 
Thus the nontrivial requirement in the definition is the 
condition that the isometry $V$ intertwines, besides the 
algebras, the respective representations 
of the Poincar\'e group.  

     Although one may expect that the warped Borchers triples are
generally inequivalent to the original ones, there does not yet
exist an  argument to that effect. It has been shown in \cite{GrLe,BuSu2} 
that in theories describing massive particles the elastic scattering matrix 
changes under these deformations, thereby providing a rather indirect proof 
that the respective Borchers triples must be inequivalent. We present here 
an alternative argument, covering a large family of theories in more than
two spacetime dimensions. It is based on the following 
lemma, whose proof is given in the appendix. There we also comment on the 
additional physically meaningful spectral constraint on the translations
made in the hypothesis.

\begin{lemma} \label{cyclic}
Let $(\Rs, U, \Omega)$ be a causal Borchers triple relative to 
$\Ws \subset \RR^n$, $n \geq 3$, such that 
$\mbox{sp} \ U \! \upharpoonright \! \RR^n$ contains some point in the 
interior of $V_+$ and let $Q \neq 0$ be a $\Ws$--admissible matrix of the
generic form (\ref{q1}). Then $\Omega$ is cyclic for at most one of the algebras 
$\bigcap_{\lambda \in \Ns}  \, \alpha_\lambda(\Rs)$
and $\bigcap_{\lambda \in \Ns}  \, \alpha_\lambda(\Rs_Q)$,
where $\Ns$ is any given neighborhood of the identity in~$\Pid$.
\end{lemma} 

     The following observation about the relation between Borchers
triples and their warped descendants is an immediate consequence of this 
result.

\pagebreak
\begin{prop} Let $(\Rs, U, \Omega)$ be a causal Borchers triple 
relative to $\Ws \subset \RR^n$, $n \geq 3$, 
such that $\mbox{sp} \ U \! \upharpoonright \! \RR^n$ contains
some point in the interior of $V_+$ and let 
$\Omega$ be cyclic for $\bigcap_{\lambda \in \Ns}  \,
\alpha_\lambda(\Rs)$ for some neighborhood $\Ns$ of the identity 
in~$\Pid$. Then  $(\Rs_Q, U, \Omega)$ and  $(\Rs, U, \Omega)$ 
are inequivalent for any $\Ws$--admissible matrix $Q \neq 0$ of the
generic form (\ref{q1}). 
\end{prop}
\begin{proof}
Let $V$ be some unitary operator which intertwines the
two triples. 
Then $V \, \bigcap_{\lambda \in \Ns} \, \alpha_\lambda(\Rs) \ V^{-1}  =  
\bigcap_{\lambda \in \Ns} \, \alpha_\lambda(V \, \Rs \,  V^{-1})  =
\bigcap_{\lambda \in \Ns}  \, \alpha_\lambda(\Rs_Q) $. Hence $\Omega =
V \Omega$ is cyclic for the latter algebra as well, in conflict with
the preceding lemma.
\end{proof}

     In the familiar examples of quantum field theories which
have been rigorously constructed so far, such as (generalized) free field 
theories in physical spacetime and interacting field theories in lower 
dimensions  \cite{GlJa}, the vacuum~$\Omega$ is known to be  
cyclic for the algebras affiliated with compact spacetime regions
(Reeh--Schlieder property). Thus, applying the warping procedure to 
the corresponding Borchers triples, one ends up with inequivalent, \ie new 
theories. However, the local algebras in the deformed theories no longer 
have the Reeh--Schlieder property, according to the preceding lemma. 
In fact, even for the algebras affiliated with pointed 
spacelike cones, which are of relevance in gauge theory \cite{BuFr},
$\Omega$ is not cyclic. Thus in more than two spacetime dimensions
the warped algebras can, in general, not be interpreted in terms of some 
underlying point fields.

     Yet, as was pointed out in \cite{BuSu2}, the warped
theories admit a meaningful physical interpretation with respect to 
noncommutative Minkowski space (Moyal space). In fact, the first examples of 
such theories appeared in that setting \cite{GrLe}. We recall that 
noncommutative Minkowski space is described by  
coordinate operators $X_\mu$ satisfying the commutation relations
$ [X_\mu , X_\nu ] = i \, \theta_{\mu \nu} \, 1$, 
where $\theta_{\mu \nu} = - \theta_{\nu \mu}$ are real constants,
$\mu, \nu = 0, 1, \dots, n-1$. It is straightforward to verify that 
in more than two dimensions there always exist certain 
lightlike coordinates $X_\pm$ which commute and thus can be simultaneously 
diagonalized. Hence it should be possible to localize fields and observables  
with respect to these coordinates, thereby dislocalizing
them in the remaining ones. In particular, the wedges $\cW$ 
considered here are possible localization regions in noncommutative 
Minkowski space, whereas bounded regions and pointed spacelike cones are not. 
On the basis of this interpretation, the algebras corresponding to the latter 
regions are expected to be trivial, in line with the preceding lemma. 
Now, apart from the wedges, there are other cylindrical regions (such as the 
intersections of opposite wedges) which are possible localization regions. 
It is therefore an intriguing question whether the corresponding 
algebras in the warped theories are nontrivial. An affirmative
answer would support their interpretation in terms of noncommutative 
Minkowski space. We hope to return to this problem elsewhere.

\section{Conclusions} \label{conclusions}
\setcounter{equation}{0}

     In this investigation we have clarified the relation between the 
warped convolution of $C^*$--dynamical systems, proposed in \cite{BuSu2}, 
and the strict deformation of such systems, established by Rieffel \cite{Ri}.
It turned out that, for fixed deformation matrix $Q$, the warped convolution 
induces a faithful covariant representation of the corresponding Rieffel 
algebra, if the original dynamical system is given in a faithful 
covariant representation. Thus, from this point of view, the warped convolution 
provides little new information. Yet, whereas the Rieffel deformations were 
introduced for the purpose of quantizing classical systems with 
Poisson bracket given by a fixed $Q$, warped convolutions were conceived for 
the deformation of quantum field theories. Within the latter framework one must
deal simultaneously with a multitude of different deformation matrices $Q$
and establish relations between the resulting operators.
The warping procedure is more appropriate in this context, since all warped 
deformations of a given dynamical system are concretely presented in a single
Hilbert space, irrespective of the choice of $Q$.

     For the discussion of the field theoretic aspects it has proven to 
be convenient to make use of the concept of causal Borchers triples 
$(\Rs, U, \Omega)$. The algebras of observables attached to arbitrary 
regions in Minkowski space can be reconstructed from any such triple, thereby 
specifying a covariant and causal quantum theory.  
Within this setting the problem of constructing a theory thus
presents itself as follows. One first has to devise  
a continuous unitary representation $U$ of the Poincar\'e
group on some Hilbert space which satisfies the relativistic 
spectrum condition with vacuum vector $\Omega$. This task can 
be accomplished, \textit{e.g.}, by specifying the stable particle 
content of the theory and performing the standard Fock space
construction. In a second step one must  exhibit a von Neumann algebra 
$\Rs$ on this space satisfying certain specific compatibility conditions 
with respect  to the action of $U$, which allow one to interpret $\Rs$ 
as an  algebra of observables localized in a given wedge region of 
Minkowski space. It should be noted that the nets of local observable 
algebras appearing in any quantum field theory can be realized in this way. 

     Disregarding systems with an unreasonably large number of local degrees 
of freedom, the algebraic structure of $\Rs$ is known to be model independent, 
\ie the algebras corresponding to different theories are isomorphic 
\cite{BuDaFr}. One may thus take as prototype the von Neumann 
algebra $\Rs_0$ generated by free (non--interacting) fields 
on Fock space which are smeared with test functions having 
support in the given wedge region. Despite this concrete setting, 
the problem of identifying other proper examples of such algebras $\Rs$ 
is notoriously difficult. The strategy pursued in the present 
investigation is based on the general idea of deforming a given causal 
Borchers triple, such as $(\Rs_0, U, \Omega)$, without changing the 
representation $U$. The warping procedure provides a consistent 
method to that effect. It leads to a 
family of new examples of causal Borchers triples in any number 
of spacetime dimensions.

\pagebreak
     However, the deformations of Borchers triples obtained by the 
warping procedure are rather special and of limited physical interest. It 
therefore seems worthwhile to fathom the potential of the general idea 
underlying this construction. Since the representation $U$ of $\Pid$ induces 
the pertinent constraints on the admissible algebras $\Rs$, one may try to 
generalize the formula for the warped deformations by the ansatz
$$
A \doteq \iint \! d\lambda \, d\lambda^\prime \, 
K(\lambda, \lambda^\prime) \, \alpha_{\lambda\lambda^\prime}(A_0) \, L(\lambda,
\lambda^\prime) \, ,  \quad A_0 \in \Rs_0 \, ,
$$
where $d\lambda$ denotes the Haar measure on $\Pid$ (or a subgroup thereof)
and $K, L$ are suitable operator valued kernels. The consistency conditions on
the algebra $\Rs$ generated by the deformed operators can then be
re-expressed in terms of transformation properties of these kernels
under the adjoint action of the representation~$U$. 

     In two spacetime dimensions these constraints simplify considerably.
There it suffices if the kernels $K, L$ transform covariantly
under the adjoint action of the unitary representation $U$  
of $\Pid$ and $\Omega$ is cyclic and separating for
the resulting deformed von Neumann algebra $\Rs$. 
One may then proceed as in Section \ref{triple} 
and extend the representation $U$ to a representation of $\Ps_+$ by
adding to it the modular conjugation associated with 
$(\Rs, \Omega)$ which can be interpreted
as PCT--operator. The algebras corresponding to
arbitrary wedges can be obtained from $\Rs$ by the adjoint 
action of the resulting (anti)unitary representation of~$\Ps_+$.
Indeed, there is evidence that a large family of 
integrable models on two--dimensional Minkowski space,
considered by one of us, can be subsumed in this manner \cite{Le3, Le4}. 

     The prospect of finding other interesting deformations of this kind 
also in higher spacetime dimensions seems promising. Moreover, the method 
can also be transferred to quantum field theories on curved spacetimes 
having a sufficiently large isometry group \cite{DLM}.
Thus the algebraic methods presented here shed new light on the 
yet unsolved constructive problems in relativistic quantum field 
theory.

\appendix
\section*{\Large Appendix }
\setcounter{equation}{0}

     We give here the proof of Lemma \ref{cyclic}, which concludes 
that, given any neighborhood $\Ns$ of the identity in $\Pid$, 
$\Omega$ is cyclic for at most one of the algebras 
$\bigcap_{\lambda \in \Ns}  \, \alpha_\lambda(\Rs)$
and $\bigcap_{\lambda \in \Ns}  \, \alpha_\lambda(\Rs_Q)$.
Moreover, we  
comment on the significance of the spectral constraint made in the
hypothesis of the lemma. 
 
We begin by noting that it suffices 
to establish the assertion for arbitrarily small neighborhoods $\Ns$
of the identity in $\Pid$; for it then holds for all 
bigger neighborhoods as well. In particular, one may
assume that $\lambda_0 \, \Ns \lambda_0^{-1} = \Ns$, where
$\lambda_0 \in \Pid$ is a rotation by $\pi$ which maps $\Ws$ onto 
$\Ws^{\, \prime}$.  Assume now that $\Omega$ is cyclic for 
$\Ss \doteq \bigcap_{\lambda \in \Ns} \, \alpha_\lambda(\Rs) \subset
\Rs$ and let $A \in \Ss \bigcap \Cs^\infty$. Then for 
any $\lambda \in \Ns$ one has $\alpha_{\lambda^{-1}}(A) \in \Rs^\infty$, so the 
warped operators ${\alpha_{\lambda^{-1}}(A)}_{Q}$ are well--defined and
$\alpha_\lambda({\alpha_{\lambda^{-1}}(A)}_{Q}) \in \alpha_\lambda(\Rs_Q)$. 
By Proposition~\ref{2.5} \
$\alpha_\lambda({\alpha_{\lambda^{-1}}(A)}_Q) = A_{\Lambda Q \Lambda^T}$,
where $\Lambda$ is the image of $\lambda$ under the canonical 
homomorphism mapping $\Pid$ onto $\Lid$. Hence 
$A_{\Lambda Q \Lambda^T} \in \alpha_\lambda(\Rs_Q)$, $\lambda \in
\Ns$. 

     Assume now, for a \textit{reductio ad absurdum}, 
that $\Omega$ is also cyclic 
for $\bigcap_{\lambda \in \Ns} \, \alpha_\lambda(\Rs_Q) $. Then, since 
$$
\big( \bigvee_{\lambda \in \Ns} \, \alpha_\lambda(\Rs_Q) \big)^\prime =
\bigcap_{\lambda \in \Ns} \, \alpha_\lambda(\Rs_Q{}^\prime)
\supset \bigcap_{\lambda \in \Ns} \,
\alpha_\lambda(\alpha_{\lambda_0}(\Rs_Q)) 
= \alpha_{\lambda_0} \big( \bigcap_{\lambda \in \Ns_0} \,
\alpha_\lambda(\Rs_Q) \big) \, ,
$$
where the inclusion obtains because $(\Rs_Q, U, \Omega)$ is a 
causal Borchers triple, one concludes that $\Omega$ is separating for 
$\bigvee_{\lambda \in \Ns} \, \alpha_\lambda(\Rs_Q)$.
But equation \eqref{vacuum} entails 
$A_{\Lambda Q \Lambda^T} \Omega = A_Q \Omega$, and consequently 
$A_{\Lambda Q \Lambda^T} = A_Q$,  $\lambda \in \Ns$. Proposition \ref{group}
then yields 
$A_{\Lambda Q \Lambda^T - Q} = A = A_{Q - \Lambda Q \Lambda^T}$,  $\lambda \in \Ns$. 
By explicit computation one finds that the sums of matrices of the form  
$\Lambda Q \Lambda^T - Q$, $\lambda \in \Ns$, include all multiples
of $Q$. Hence $A_{mQ} = A$, $m \in \ZZ$, by another 
application of Proposition~\ref{group}.  
The same is true for the smooth operators 
$A^\prime \in \Ts \doteq \alpha_{\lambda_0} (\Ss) \subset \alpha_{\lambda_0} (\Rs)
\subset \Rs^\prime$, as one sees by applying again Proposition~\ref{2.5}. 

     Pick now an arbitrary compact subset $\Gamma$ in the
interior of the forward lightcone $V_+$, so $Q \, \Gamma$ is 
a compact subset in the interior of $\Ws$. Hence for
any given $x,y \in \RR^n$ and sufficiently large $m \in \NN$,
the wedges $\Ws + x + m \, Q u$ and $\Ws^\prime + y - m \, Q v$ lie spacelike
to each other for all $u \in \Gamma$ and $v \in V_+$. 
As explained in Section \ref{application}, one therefore has for any 
$A \in \Rs$, $\As^\prime \in \Rs^\prime$, the equality 
$ [ \alpha_{x + m \, Qu} (A), \alpha_{y - m \, Qv}(A^\prime) ] = 0$.
Now let $A \in \Ss \bigcap \Cs^\infty$, 
$A^\prime \in \Ts \bigcap \Cs^\infty$, let $\Phi$ be 
any vector with spectral support with respect to 
$U \upharpoonright \RR^n$ contained in $\Gamma$, and let $\Psi$ be any other
vector with compact spectral support. According to 
the preceding step and Proposition~\ref{2.5}  one has 
$\alpha_x(A) = \alpha_x(A_{m Q}) = (\alpha_x(A))_{mQ}$ and
similarly $\alpha_y(A^\prime) = (\alpha_y(A^\prime))_{-mQ}$. So
one obtains by the same line of arguments as in the proof of 
Proposition \ref{2.4}, 
\begin{align*}
& \langle \Phi, \, \alpha_x(A) \,  \alpha_y(A^\prime) \, \Psi \rangle 
=  \lim_{m \rightarrow \infty}
\langle \Phi, (\alpha_x(A))_{m Q} \,  
(\alpha_y(A^\prime))_{- m Q} \, \Psi \rangle \\
&  = \lim_{m \rightarrow \infty}
\langle \Phi, (\alpha_y(A^\prime))_{-m Q} \,     
(\alpha_x(A))_{m Q} \, \Psi \rangle  
= \langle \Phi, \, \alpha_y(A^\prime) \,  \alpha_x(A) \, \Psi \rangle 
\, ,
\end{align*}
where in the second equality the support properties of $\Phi$ and
the above commutation properties of $A, A^\prime$ have been used.
Thus, varying $\Phi, \Psi$ within the above limitations, one arrives at 
\begin{equation*} \label{triviality}
E(\Gamma) \, [\alpha_x(A), \alpha_y(A^\prime)] = 0 \quad
\mbox{for} \quad  x,y \in \RR^n \, ,  
\end{equation*}
where $E(\, \cdot \,)$ 
denotes the spectral resolution of $U \upharpoonright \RR^n$. 

   This equality has been established for 
$A \in \Ss \bigcap \Cs^\infty$ and $A^\prime \in \Ts \bigcap \Cs^\infty$.
But if $\Ns$ is sufficiently small, the algebra ${\Ss}$
is mapped into itself by all translations in the open convex 
cone $ \bigcap_{\lambda \in \Ns} \Lambda \Ws $;
appealing to the discussion following relation~(\ref{smoothing}) 
allows one to conclude that  $\Ss \bigcap \Cs^\infty$ is weakly dense
in ${\Ss}$ and, similarly,   $\Ts \bigcap \Cs^\infty$ is weakly dense
in ${\Ts}$.  So the equality holds for all 
$A \in \Ss$ and $A^\prime \in \Ts$. Moreover, for any 
$u,v \in \RR^n$ there is a  $w \in \RR^n$ such that 
$\alpha_w(\Ts) \supset \alpha_u(\Ts) \bigvee \alpha_v(\Ts)$.
(This follows from the the Poincar\'e covariance discussed in Section 
\ref{application} and the geometry of wedge regions).
Hence 
$E(\Gamma) \, [A, T] = 0 $ for $A \in \Ss$ 
and $T \in \bigvee_{y \in \RR^n} \alpha_y(\Ts)$.

  Since $\Omega$ is cyclic for ${\Ss}$ 
it is also cyclic  for ${\Ts} = \alpha_{\lambda_0}(\Ss)$. 
The spectral condition on $U \upharpoonright \RR^n$ therefore 
implies that 
the elements of $\bigcap_{\, y \in \RR^n} \alpha_y(\Ts)^\prime$
are invariant under translations. In particular 
$U(x) \in \bigvee_{y \in \RR^n} \alpha_y(\Ts)$, $x \in \RR^n$, 
\cf \cite[Theorem~4.6]{Ar}. Thus  
$E(\Gamma) \, [A, U(x)] = 0 $,  $x \in \RR^n$, and consequently  
$E(\Gamma) \, A \Omega = 0$, $A \in \Ss$.  
It is then clear that $E(\Gamma) = 0$ 
for any compact subset $\Gamma$ in the interior of $V_+$. So the 
spectrum of $U \upharpoonright \RR^n$ is confined to the boundary of 
the lightcone $V_+$, \ie there is no spectral point in its 
interior, contradicting the hypothesis of the lemma. 
This completes the proof of the lemma.

     Finally, let us discuss the significance of the assumption that
the spectrum of $U \upharpoonright \RR^n $ intersects the interior 
of $V_+$. As a matter of fact, disregarding the trivial case 
$\mbox{sp} \, U = \{ 0 \}$, this input is a consequence of the 
additivity of the energy--momentum spectrum, which can be established in 
the present setting if $\Omega$ is (apart from a phase) the only unit vector 
in the underlying Hilbert space which is invariant under translations 
\cite[Chapter II.5.4]{Haag}. The possibility that $\mbox{sp} \, U$ consists 
of the boundary of  $V_+$ (and thus is not additive) can only be realized in 
theories where the Lorentz symmetry is spontaneously broken. With the help 
of one--dimensional chiral fields which one assigns to lightrays, one can 
manufacture such examples, and these are stable under the warping procedure. 
Since these examples seem to be merely of academic interest, 
we do not present them here. 

\vspace*{7mm} 
\noindent \textbf{\Large Acknowledgments}  \\[2mm] 
{\small GL wishes to thank S.~Waldmann for interesting 
discussions about Rieffel deformations. }



\begin{thebibliography}{99}

\footnotesize 

\bibitem{Ar}
H. Araki, {\it Mathematical Theory of Quantum Fields} (Oxford University 
Press, Oxford) 1999.

\bibitem{BaWo}
H.\ Baumg\"artel and M.\ Wollenberg, {\it Causal Nets of Operator
  Algebras}, (Akademie Verlag, Berlin) 1992



\bibitem{Bo}
H.-J. Borchers, The CPT-theorem in two-dimensional theories of local 
observables, {\sl Commun. Math. Phys., \bf 143}, 315--332 (1992).

\bibitem{Bo1}
H.-J. Borchers, 
On revolutionizing quantum field theory with Tomita???s modular theory,
{\sl J. Math. Phys. \bf 41}, 3604--3673 (2000).



\bibitem{BrGuLo2}
R. Brunetti, D. Guido and R. Longo, Modular localization and Wigner
particles, {\sl Rev. Math. Phys., \bf 14}, 759--785 (2002).



\bibitem{BuDaFr}
D. Buchholz, C. D'Antoni  and K. Fredenhagen,
The universal structure of local algebras, 
{\sl Commun. Math. Phys. \bf 111}, 123 (1987) 

\bibitem{BDFS}
D. Buchholz, O. Dreyer, M. Florig and S.J. Summers, Geometric 
modular action and spacetime symmetry groups, {\sl Rev. Math. Phys., \bf 12}, 
475--560 (2000).

\bibitem{BuFr}
D. Buchholz and K. Fredenhagen, Locality and the structure of
particle states, {\sl Commun. Math. Phys., \bf 84}, 1--54 (1982).

\bibitem{BuLe}
D. Buchholz and G. Lechner, Modular nuclearity and localization,
{\sl Ann. Henri Poincar\'e, \bf 5}, 1065--1080 (2004).

\bibitem{BuSuIII}
D. Buchholz and S.J. Summers, 
An Algebraic characterization of vacuum states in Minkowski space. 3. Reflection maps.
{\sl Commun. Math. Phys. \bf 246}, 625--641 (2004). 

\bibitem{BuSuads}
D. Buchholz and S.J. Summers, Stable quantum systems in Anti-de Sitter
space: Causality, independence and spectral properties, {\sl J. Math. Phys., 
\bf 45}, 4810--4831 (2004).

\bibitem{BuSu1}
D. Buchholz and S.J. Summers, String-- and brane--localized causal
fields in a strongly nonlocal model, {\sl J. Phys. A, \bf 40},
2147--2163 (2007).

\bibitem{BuSu2}
D. Buchholz and S.J. Summers, Warped convolutions: A novel tool 
in the construction of quantum field theories, in: {\it Quantum Field
Theory and Beyond}, edited by E. Seiler and K. Sibold (World Scientific,
Singapore), pp. 107--121, 2008.

\bibitem{DLM}
C.~Dappiaggi, G.~Lechner, E.~Morfa-Morales, {Article in preparation} 



\bibitem{Fl}
M. Florig, On Borchers' theorem, {\sl Lett. Math. Phys., \bf 46}, 
289--293 (1998).

\bibitem{GlJa}
J. Glimm and A. Jaffe,
\textit{Quantum Physics. A Functional Integral Point of View},  
(Springer Verlag, Berlin, Heidelberg and New York)  1987.

\bibitem{GrLe}
H. Grosse and G. Lechner, Wedge--local quantum fields and
noncommutative Minkowski space, {\sl JHEP}, {\bf 0711}, 012 (2007). 

\bibitem{GrLe2}
H. Grosse and G. Lechner, Noncommutative deformations of Wightman
quantum field theories, {\sl JHEP}, {\bf 0809}, 131 (2008).

\bibitem{Gu}
D. Guido, Modular covariance, PCT, Spin and Statistics, {\sl Ann.
Inst. Henri Poincar\'e, \bf 63}, 383--398 (1995).

\bibitem{Haag}
R. Haag, {\it Local Quantum Physics}, (Springer Verlag, Berlin, Heidelberg 
and New York) 1992. 


\bibitem{KNW}  
D. Kaschek, N. Neumaier and S. Waldmann, Complete positivity of Rieffel's
quantization by actions of $\RR^d$, {\sl J. Noncommut. Geom., \bf 3},
361--375 (2009). 

\bibitem{Le}
G. Lechner, Polarization-free quantum fields and interaction,
{\sl Lett. Math. Phys., \bf 64}, 137--154 (2003).

\bibitem{Le2}
G. Lechner, On the existence of local observables in theories
with a factorizing S-matrix, {\sl J. Phys. A, \bf 38}, 3045--3056 (2005).

\bibitem{Le3}
G. Lechner, Construction of quantum field theories with factorizing
S-matrices, {\sl Commun. Math. Phys.}, {\bf 277}, 821--860 (2008). 

\bibitem{Le4}
G. Lechner, {Article in preparation} 


\bibitem{MuSchYng}
J. Mund, B. Schroer and J. Yngvason, String--localized quantum fields
and modular localization, {\sl Commun. Math. Phys., \bf 268}, 
621--672 (2006).

\bibitem{Pe}
G.K. Pedersen, {\it C$^*$--Algebras and Their Automorphism Groups}, 
(Academic Press, London, New York and San Francisco) 1979.


\bibitem{Ri}
M.A. Rieffel, Deformation quantization for actions of $\RR^d$, 
{\sl Memoirs A.M.S., \bf 506}, 1--96 (1993).


\bibitem{Sch}
B. Schroer, Modular localization and the bootstrap--formfactor
program, {\sl Nucl. Phys. B, \bf 499}, 547--568 (1997).



\bibitem{Ta}
M. Takesaki, {\it Tomita's Theory of Modular Hilbert Algebras and Its
Applications}, (Springer Verlag, Berlin, Heidelberg and New York) 1970.

\bibitem{Ta2}
M. Takesaki, {\it Theory of Operator Algebras}, Volume II,
(Springer Verlag, Berlin, Heidelberg and New York) 2003.




\end{thebibliography}
\end{document}